\documentclass[journal,twoside]{IEEEtran}

\usepackage{amsmath,amsfonts}
\usepackage{textcomp}
\usepackage{ragged2e}
\usepackage{stfloats}
\usepackage{url}
\usepackage{verbatim}
\usepackage[normalem]{ulem}
\usepackage{blindtext}
\usepackage{pdfpages}
\usepackage{multicol}
\usepackage[utf8]{inputenc}
\usepackage{environ}         
\usepackage{graphicx}        
\usepackage{graphics}
\usepackage{rotating}
\usepackage{setspace}
\usepackage{dsfont}
\usepackage{xcolor}
\usepackage[T1]{fontenc}
\usepackage{mathtools}
\usepackage{pifont}
\usepackage{makecell}
\usepackage{stackengine}
\usepackage{bm}
\usepackage{bbm}
\usepackage{array}
\usepackage{tabu}
\usepackage{multirow}
\usepackage[linesnumbered,lined,boxed,commentsnumbered,ruled]{algorithm2e}
\usepackage{float}
\usepackage{subfloat}
\usepackage{caption}
\usepackage{subcaption}
\usepackage{soul}
\usepackage{amsthm,amssymb}
\usepackage{mathrsfs}
\usepackage{cite}
\usepackage{xfrac}
\usepackage[font={small}]{caption}
\usepackage{ptext}
\usepackage{tikzsymbols}
\usepackage{etoolbox}
\usepackage[colorlinks=true, linkcolor    = black, urlcolor     = black, citecolor=black, pdfborder={0 0 0}]{hyperref}
\usepackage{fontawesome5}
\usepackage{academicons}
\usepackage{xcolor}
\newcommand{\orcid}[1]{\href{https://orcid.org/#1}{\textcolor[HTML]{A6CE39}{\faOrcid}}}

\DeclarePairedDelimiterX\Basics[1](){ #1}

\setcellgapes{1.5pt}
\DeclareMathAlphabet{\mathpzc}{OT1}{pzc}{m}{it}

\DeclarePairedDelimiter\bceil{\big\lceil}{\big\rceil}

\theoremstyle{remark}
\newtheorem{definition}{Definition}

\newtheorem{example}{Example}

\newtheorem{proposition}{Proposition}
\newtheorem{lemma}{Lemma}
\newtheorem{remark}{Remark}

\newcommand{\alabel}[1]{\stepcounter{equation}\tag{\theequation}\label{#1}}

\newcommand{\matr}[1]{\mathbf{#1}}
\newcommand{\matrA}{\matr{A}}

\newcommand{\matrQ}{\matr{Q}}
\newcommand{\tmatrQ}{\tilde{\matr{Q}}}
\newcommand{\cmatrQ}{\matr{\check Q}}

\newcommand{\vect}[1]{\mathbf{#1}}
\newcommand{\vS}{\vect{S}}

\newcommand{\vX}{\vect{X}}
\newcommand{\vx}{\vect{x}}
\newcommand{\cvx}{\vect{\check x}}
\newcommand{\vY}{\vect{Y}}
\newcommand{\vy}{\vect{y}}
\newcommand{\cvy}{\vect{\check y}}
\newcommand{\vZ}{\vect{Z}}
\newcommand{\vz}{\vect{z}}
\newcommand{\cvz}{\vect{\check z}}

\newcommand{\Es}{E_{\mathrm{s}}}

\newcommand{\vgamma}{\boldsymbol{\gamma}}

\newcommand{\cgamma}{\boldsymbol{\check \gamma}}

\newcommand{\set}[1]{\mathcal{#1}}
\newcommand{\setB}{\set{B}}

\newcommand{\setM}{\set{M}}
\newcommand{\setS}{\set{S}}

\newcommand{\setQ}{\set{Q}}
\newcommand{\setX}{\set{X}}
\newcommand{\setY}{\set{Y}}
\newcommand{\setZ}{\set{Z}}

\newcommand{\setSleft}[1]{\overleftarrow{\setS}_{\! #1}}
\newcommand{\setSright}[1]{\overrightarrow{\setS}_{\! #1}}

\newcommand{\tp}{\tilde{p}}
\newcommand{\hQ}{\hat{Q}}
\newcommand{\tQ}{\tilde{Q}}
\newcommand{\tT}{\tilde{T}}
\newcommand{\tmu}{\tilde{\mu}}
\newcommand{\hmu}{\hat{\mu}}
\newcommand{\ttheta}{\tilde{\theta}}
\newcommand{\hp}{\hat{p}}
\newcommand{\chp}{\check{\hat{p}}}
\newcommand{\chQ}{\check{\hat{Q}}}

\newcommand{\gB}{g^{\userB}}
\newcommand{\gE}{g^{\userE}}
\newcommand{\tgE}{\tilde g^{\userE}}

\newcommand{\gBt}{g^{\userB}_t}
\newcommand{\gEt}{g^{\userE}_t}

\newcommand{\mB}{m_{\userB}}
\newcommand{\mE}{m_{\userE}}

\newcommand{\sigmaB}{\sigma_{\userB}}
\newcommand{\sigmaE}{\sigma_{\userE}}

\newcommand{\NBidx}[1]{N^{\userB}_{#1}}

\newcommand{\NEidx}[1]{N^{\userE}_{#1}}

\newcommand{\Rs}{R_{\mathrm{s}}}
\newcommand{\cRs}{\check R_{\mathrm{s}}}

\newcommand{\bpsi}{\bar{\psi}}

\newcommand{\deltaQ}{\delta Q}
\newcommand{\deltamu}{\delta \mu}

\newcommand{\rdiff}[1]{\frac{\operatorname{d}}%
	{\operatorname{d} \hskip-0.035cm #1}}

\newcommand{\userB}{\mathrm{B}}
\newcommand{\userE}{\mathrm{E}}

\newcommand{\TB}{T^{\userB}}
\newcommand{\TE}{T^{\userE}}
\newcommand{\cTB}{\check T^{\userB}}
\newcommand{\cTE}{\check T^{\userE}}

\newcommand{\tTB}{\tT^{\userB}}
\newcommand{\tTE}{\tT^{\userE}}

\newcommand{\SNRB}{\mathrm{SNR}^{\userB}}
\newcommand{\SNRE}{\mathrm{SNR}^{\userE}}

\newcommand{\SNRdBB}{\mathrm{SNR}^{\userB}_{\textrm{dB}}}
\newcommand{\SNRdBE}{\mathrm{SNR}^{\userE}_{\textrm{dB}}}

\newcommand{\defeq}{\triangleq}

\newcommand{\Z}{\mathbb{Z}}
\newcommand{\R}{\mathbb{R}}
\newcommand{\C}{\mathbb{C}}

\newcommand{\plimsup}[1]{\operatorname{p-}\!\limsup_{#1}}
\newcommand{\pliminf}[1]{\operatorname{p-}\!\liminf_{#1}}

\newcommand{\Prob}{\operatorname{Pr}}

\newcommand{\defend}{\mbox{}\hfill$\square$}

\newcommand{\exampleend}{\mbox{}\hfill$\square$}
\newcommand{\remarkend}{\mbox{}\hfill$\square$}

\definecolor{dukeblue}{rgb}{0.0, 0.0, 0.61}
\definecolor{harvardcrimson}{rgb}{0.79, 0.0, 0.09}

\newcommand{\setMn}{\setM}
\newcommand{\setMnsq}{\setM^2}
\newcommand{\Mn}{M}
\newcommand{\mn}{m}

\newcommand{\tc}{t_{\mathrm{c}}}

\newcommand{\bigformulatop}[3]{%
	\begin{figure*}[!t]
		\normalsize
		\setcounter{equation}{#1}
		#3
		
		\setcounter{equation}{#2}
		\hrulefill
		\vspace*{4pt}
	\end{figure*}
}

\newcommand{\bigformulabottom}[3]{%
	\begin{figure*}[!b]
		\normalsize
		\vspace*{4pt}
		\hrulefill
		
		\setcounter{equation}{#1}
		#3
		
		\setcounter{equation}{#2}
	\end{figure*}
}

\begin{document}


\title{Constrained Secrecy Capacity of Finite-Input
	Intersymbol Interference Wiretap Channels}


\author{Aria~Nouri\textsuperscript{\,\orcid{0000-0001-5548-184X}},~\IEEEmembership{Graduate Student Member,~IEEE,}
	Reza~Asvadi\textsuperscript{\,\orcid{0000-0001-9898-7744}},~\IEEEmembership{Senior Member,~IEEE,} \\
	Jun~Chen\textsuperscript{\,\orcid{0000-0002-8084-9332}},~\IEEEmembership{Senior Member,~IEEE,}
	and Pascal~O.~Vontobel\textsuperscript{\,\orcid{0000-0002-6180-258X}},~\IEEEmembership{Fellow,~IEEE}
	\thanks{A.\ Nouri and R.\ Asvadi are with 
		the Department of Telecommunications,
		Faculty of Electrical Engineering,
		Shahid Beheshti University, Tehran 19839 69411,
		Iran
		(e-mails: ariya@ieee.org; r\_asvadi@sbu.ac.ir).
		R.\ Asvadi is the corresponding author.}%
	\thanks{J.\ Chen is with the 
		Department of Electrical and Computer Engineering,
		McMaster University, Hamilton, ON L8S 4K1, Canada
		(e-mail: chenjun@mcmaster.ca).}%
	\thanks{P.~O.\ Vontobel is with the 
		Department of Information Engineering and 
		the Institute of Theoretical Computer Science and Communications,
		The Chinese University of Hong Kong, Hong Kong SAR
		(e-mail: pascal.vontobel@ieee.org).}%
	\thanks{This paper was presented in part at the
		IEEE Information Theory Workshop, Kanazawa, Japan, 
		Oct. 2021~\cite{Nouri:Asvadi:Chen:Vontobel:21:1}.}%
}

\markboth{IEEE Transactions on Communications}%
{Nouri \MakeLowercase{\textit{et al.}}: Constrained Secrecy Capacity of Finite-Input
	Intersymbol Interference Wiretap Channels}


\maketitle

\begin{abstract}
	We consider reliable and secure communication
	over intersymbol interference wiretap channels (ISI-WTCs). In
	particular, we first derive an achievable secure rate for ISI-WTCs without
	imposing any constraints on the input distribution. Afterwards, we focus
	on the setup where the input distribution of the ISI-WTC is constrained to
	be a time-invariant finite-order Markov chain. Optimizing the parameters
	of this Markov chain toward maximizing the achievable secure rates is
	a computationally intractable problem in general, and so, toward finding a
	local maximum, we propose an iterative algorithm that at every iteration
	replaces the secure rate function with a suitable~surrogate function whose
	maximum can be found efficiently.  Although the secure rates achieved in the
	unconstrained setup are potentially larger than the secure rates achieved in
	the constrained setup, the latter setup has the advantage of leading to
	efficient algorithms for estimating and optimizing the achievable secure rates, and also has the benefit of being the basis of efficient coding schemes.
\end{abstract}

\begin{IEEEkeywords}
	Intersymbol interference (ISI),
	intersymbol interference wiretap channel (ISI-WTC),
	finite-state machine channel (FSMC),
	Markov source,
	secure rate,
	expectation-maximization (EM).
\end{IEEEkeywords}

\section{Introduction}
\label{sec::Intro}


\subsection{Motivation}


\IEEEPARstart{T}{he} increasing number of connected users and the broadcasting
nature of the wireless medium lead to a flurry of security challenges for
wireless communication applications. For example, typical cryptographic
protocols require significant communication resources for distributing and
maintaining secret keys. This issue noticeably decreases the data transmission
efficiency as the number of users gradually increases~\cite{7539590}. In
addition, traditional cryptosystems rely on the assumption that eavesdroppers
have limited computational power, making them vulnerable against more and more
powerful (quantum) computers~\cite{Gidney_2021}. Alternatively,
information-theoretic secrecy~\cite{9380147} utilizes the inherent randomness
of communication channels to achieve security
at the physical layer~\cite{8509094} without requiring secret key agreement
and without imposing any constraints on the eavesdroppers' computational
power.


The emergence of different wireless applications gives rise to diverse channel
models for various channel conditions.
Intersymbol interference (ISI)
channels are used as a model for high-data-rate transmission over wireless channels
when the delay spread of the channel exceeds the symbol duration~\cite[Ch. 9]{Proakis2008}.
In order to be specific, consider a
multipath fading channel
\begin{equation*}
	Y(\tc) \defeq \sum_{\ell=0}^{m_{\mathrm{c}}} g_{\mathrm{c},\ell}(\tc) X(\tc - \tau_\ell) + N(\tc),
\end{equation*}
with continuous-time variable $\tc \in \R$, where $X(\tc)$, $Y(\tc)$, and $N(\tc)$
denote the input, the output, and the additive noise signal, and where
$g_{\mathrm{c},\ell}(\tc)$ and $\tau_\ell$ are, respectively, the gain and the delay of the
$\ell$-th path, $0 \leq \ell \leq m_{\mathrm{c}}$. When the symbol duration is smaller than $\tau_\ell-\tau_0$, the sampled output of a
filter matched to the shaping pulse at the receiver leads to the ISI phenomenon. Such an ISI model usually appears in
single-carrier communication systems, which require a higher power efficiency
and a better peak-to-average power ratio (compared with multicarrier
communication systems) and which appear in applications of the narrowband
internet of things (NB-IoT)\footnote{In typical applications of the NB-IoT, ISI is mitigated by appending a sufficiently large cyclic~prefix to each transmitted block~\cite{7931557}. This method decreases the effective throughput as the delay spread of the channel increases.}~\cite{8698792} as outlined in specifications of 5G
and beyond-5G networks~\cite{3GPPTS36211,ericssonwpp:IoT}. Note that ISI is
also caused by multipath propagation in long-range underwater acoustic
communications~\cite{9246302}, as well as in high data-rate ultra-wideband
	communication systems~\cite{1569979}.



Providing security at the physical layer of the above-mentioned communication
technologies without imposing extra delay, power consumption, and processing
burden, has received significant attention recently
\cite{8856252,8715341,8428404}. In this paper, we mostly use scenarios from the NB-IoT
	technology for our examples and simulations. It is worthwhile to note that the NB-IoT
mostly inherits the long-term evolution (LTE)
infrastructure~\cite{3GPPTS36211}, so the essential channels operate in
the licensed sub-GHz spectrum range~\cite{3GPPTS36101}. In this spectrum
range, in contrast to the broadcasting applications in the THz carrier
frequency range~\cite{Ma-2018}, one cannot choose a sufficiently narrow angular divergence for
	the transmitter beam  toward preventing the
	eavesdropper from intercepting non-line-of-sight transmission signals.
This issue presents a vulnerable environment at the physical layer of
applications using the NB-IoT.

These considerations motivate us to study theoretical aspects of the
physical layer security over ISI wiretap channels (ISI-WTCs).
As~depicted~in~Fig.~\ref{FIG:BLK_diag}, the ISI-WTC comprises
two ISI channels, where the primary channel connects a transmitter (called
Alice) to a legitimate receiver (called Bob and abbreviated by ``$\userB$''),
while the secondary channel connects the transmitter to an eavesdropper
(called Eve and abbreviated by ``$\userE$''). In order to focus on the key
aspects of this setup, the channel gains are assumed to be constant and
perfectly known to the receiver over each transmission block.\footnote{These
	assumptions are well established in slowly-varying channels and appear also
	in other studies of ISI channels (see, e.g.,~\cite{8861076}).}


\begin{figure}
	\centering
	\includegraphics[scale=0.8]{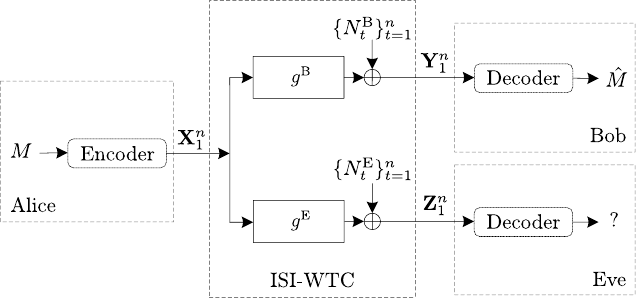}
	\caption{Block diagram of the ISI-WTC. \\ \mbox{}}
	\label{FIG:BLK_diag}
\end{figure}


\subsection{Background, Related Works, and Contributions}


ISI channels with finite memory length and finite input alphabets are a
particular case of finite-state machine channels
(FSMCs)~\cite{Gallager:1968:ITR:578869}. Toward maximizing the achievable
information rates over FSMCs, the classical Blahut-Arimoto algorithm
(BAA)~\cite{1054855,1054753} was generalized in~\cite{4494705} to optimize
finite-state machine sources (FSMSs) at the input of FSMCs. Comparing lower
bounds on the capacity of FSMCs (i.e., the maximized information
rates)~\cite{965977,4494705} with the corresponding upper
bounds~\cite{1397923,955166} typically shows a small gap between them,
which can be further narrowed by increasing the memory order of the
FSMS at the~input~\cite{4455735}.


Recently, Han and Sasaki~\cite{8747418,9483917} derived the secrecy capacity
of memoryless wiretap channels with channel state
information at the encoder. Dai \emph{et~al}.~\cite{8963770} applied these
results to physically degraded Gaussian wiretap channels with noiseless
private feedback from Bob's observations to the encoder. It was shown
in~\cite{8963770} that the considered feedback enhances the secrecy capacity
under the weak secrecy criterion. The delayed version of this feedback is
employed in~\cite{7774989} to enlarge the rate-equivocation region of
finite-state Markov wiretap channels.\footnote{A finite-state Markov wiretap
channel, as in~\cite{7774989}, is a wiretap channel where Bob's channel and
Eve's channel are FSMCs where the (joint) state process is assumed to be a
stationary ergodic Markov chain independent of the transmitted message.}
Besides employing the feedback channel, the
efficiency of secure communication over ISI channels can be enhanced
by injecting cooperative artificial noise toward degrading Eve's channel while
minimizing the impact on Bob's channel, as done in~\cite{8464866} and \cite{8737786}.


Due to power efficiency requirements, artificial-noise-aided communication
has not received too much attention in recent technologies.  Also,
establishing a private noiseless channel to feed back the complete output of
Bob's channel to Alice's encoder imposes a tremendous delay and processing
overload on the higher layers of large cooperative networks. Hence, we
focus on the standard version of ISI-WTCs (with neither feedback nor
additional artificial noise), as it requires few assumptions and consequently
is more practically relevant.


In terms of the main focus of this paper, estimating the secrecy capacity of a
finite-state wiretap channel was already considered
in~\cite{5351376}.\footnote{Note that in~\cite{5351376} a finite-state wiretap
	channel is defined to be a wiretap channel where Bob and Eve observe the
	input source through two distinct FSMCs.} However, the channel setup and the
approach to estimate the secrecy capacity in~\cite{5351376} have the following
limitations. Firstly, the assumptions for the channel setup
in~\cite{5351376} resemble the general assumptions for memoryless
wiretap channels as in~\cite{1055892}, where Eve's channel is assumed to be
noisier than Bob's channel. However, as we will show, these assumptions are
inadequate for ISI channels (and more generally, for FSMCs) due to the
non-flat frequency responses of these channels. Secondly, the gradient of the
	function that is used for approximating the secure rate function is usually not
	the same as the gradient of the secure rate function
	at a given operating point. This issue leads to an inaccurate search direction
and eventually makes the algorithm unstable.


In the following, we highlight the main contributions and results presented in
this paper.
\begin{itemize}
	
	\item 
	In the first step, we derive the achievable secure rates without imposing any
	constraints on the input distribution. We then focus on the setup where
		the input distribution is a time-invariant finite-order Markov chain,
		henceforth called an input Markov source.\footnote{Throughout the
			paper, when we talk about a Markov source at the input of the channel,
			we refer to the Markov source that models the statistics of the
			codebook~\cite{Nouri:Asvadi:22:ISIT}. It should not be confused with the
			source generating the data that we want to transmit reliably and
			securely.
		} Note that employing Markov sources at the input of the ISI
	channels has the benefit of leading to efficient algorithms for
	estimation~\cite{1661831} and maximization~\cite{4494705} of information
	rates, approaching the capacity in point-to-point setups~\cite{4455735}, and
	being a basis for efficient encoding and decoding
	schemes~\cite{1397934,Nouri:Asvadi:22:ISIT}.  Accordingly, we propose an
	efficient algorithm for optimizing the parameters of an input Markov
	source toward maximizing the obtained achievable secure rates over
	ISI-WTCs.
	
	\item Maximizing the above-mentioned secure rate is challenging because it is
	not a closed-form function of the input distribution and its
	evaluation is only possible through Monte-Carlo simulations. The key idea
	behind the proposed algorithm is to iteratively approximate the
		zeroth-order and the first-order behavior of the secure rate function by
		suitable surrogate functions that are well-defined and can relatively
		\text{easily be maximized}.
	
	\item We provide examples where the capacity of Eve's channel is higher than
	the capacity of Bob's channel, yet a nonzero secure rate is possible.  These
	examples show that it is feasible to optimize an input Markov source
	such that spectral discrepancies between the frequency responses of Bob's
	and Eve's channels can be exploited---without any further power consumption
	for transmitting interfering artificial noise toward jamming
	Eve's~channel.
\end{itemize}


\subsection{Paper Organization}

The remainder of this paper is organized as follows. Section~\ref{sec::Pre}
introduces the system model and some preliminary concepts related to FSMCs,
ISI-WTCs, and achievable secure rates. Section~\ref{sec::optim} describes
the proposed algorithm for optimizing the parameters of a Markov source at the input of an ISI-WTC and
analyzes it in detail. Section~\ref{sec::simul} contains some numerical results and
discussions. Finally, Section~\ref{sec::conclu} draws the conclusions.


\subsection{Notation}


The sets of integers and complex numbers are denoted by $\Z$ and $\C$,
respectively. The ring of polynomials with coefficients in $\C$ and
indeterminate $D$ is denoted by $\C[D]$, where ``$D$'' stands for ``delay''. Other than
that, sets are denoted by calligraphic letters, e.g.,
$\setS$. The Cartesian product of two sets $\setX$ and $\setY$ is written as
$\setX \times \setY$, and the $n$-fold Cartesian product of $\setX$ with
itself is written as $\setX^n$.  If $\setX$ is a finite set, then its
cardinality is denoted~by~$|\setX|$.


Random variables are denoted by upper-case italic letters, e.g., $X$, their
realizations by the corresponding lower-case letters, e.g., $x$, and the set
of possible values by the corresponding calligraphic letter, e.g.,
$\set{X}$. Random vectors are denoted by upper-case boldface letters, e.g.,
$\vX$, and their realizations by the corresponding lower-case letters, e.g.,
$\vx$. For integers $n_1$ and $n_2$ satisfying $n_1 < n_2$, the notation
$\vX_{n_1}^{n_2} \defeq (X_{n_1},X_{n_1+1},\ldots,X_{n_2}) $ is used for a
time-indexed vector of random variables and
$\vx_{n_1}^{n_2}\defeq(x_{n_1},x_{n_1+1},\ldots,x_{n_2})$ for its realization.
Boldface letters are also used for matrices, e.g., $\matrA$, with the
$(i,j)$-entry of $\matrA$ being called $A_{ij}$.




%


For any real number $x$, the expression $(x)^+$ stands for $\max\{x,0\}$;
similarly, $f^+(\cdot)$ stands for $\big(f(\cdot)\big)^+$. Moreover, the
expression $\log(\,\cdot\,)$ denotes the natural logarithm function.


The entropy of a random variable $X$,  the mutual information between
two random variables $X$ and $Y$, and the mutual information between
two random variables $X$ and $Y$ conditioned on the random variable
$Z$ are denoted by $H(X)$, $I(X;Y)$, and $I(X;Y|Z)$, respectively.
Finally, the variational distance between the probability mass functions
(PMFs) of two random variables $X$ and $Y$ over the same finite alphabet
$\setX$ is defined as 
$$
d_{\setX}(p_X,p_Y) \defeq \sum\limits_{x \in\setX} \ \bigl|
p_X(x) - p_Y(x) \bigr|.
$$


\section{Preliminaries}
\label{sec::Pre}

\subsection{Channel Model}

An ISI channel with transfer polynomial
$g(D) \defeq \sum_{t=0}^{m} g_t D^t \in \C[D]$, where $g_m\neq 0$ and where
$m$ is called the memory length, has an input process
$\{ X_t \}_{t \in \Z}$, a noiseless output process $\{ U_t \}_{t \in \Z}$, a noise process
$\{ N_t \}_{t \in \Z}$, and a noisy output process $\{ Y_t \}_{t \in \Z}$
with
\begin{alignat*}{2}
    U_t
      &\defeq
         \sum_{\ell=0}^{m} g_{\ell} X_{t-\ell},
           &&\quad t \in \Z, \\
    Y_t
      &\defeq
         U_t + N_t,
           &&\quad t \in \Z,
\end{alignat*}
where $X_t, U_t, {N_t,\,} Y_t \in \C$ for all $t \in \Z$, and where
the noise process is independent of the channel input process. In the
following, we will assume that the noise process is white Gaussian noise,
i.e., $\{ N_t \}_{t \in \Z}$ are i.i.d.\ Gaussian random variables with mean
zero and variance~$\sigma^2$. Clearly, an ISI channel is parameterized by
the couple $\bigl( g(D), \sigma^2 \bigr)$.

An ISI channel described by
$\bigl( g(D) \defeq \sum_{t=0}^{m} g_t D^t, \ \sigma^2 \bigr)$, where
$m < \infty$, and having an input process $\{ X_t \}_{t \in \Z}$ taking values
in a finite set $\setX \subsetneq \C$ is a special case of the channels
	in the class of finite-state machine channels (FSMCs), which were called
finite-state channels in~\cite{Gallager:1968:ITR:578869}. Indeed,~let 
$s_t \defeq \vx_{t-\nu+1}^t$ (with $\setS \defeq \setX^{\nu}$ and $\nu\geq m$) denote the state of an
FSMC modeling an ISI channel at $t\in\Z$. Then
\begin{align*}  
	&
	p_{S_t, Y_t | S_{t-1}, X_t}(s_t, y_t | s_{t-1}, x_t)\\
	&\quad\quad
	= p_{S_t | S_{t-1}, X_t}(s_t | s_{t-1}, x_t)
	\cdot
	p_{Y_t | S_{t-1}, X_t}(y_t | s_{t-1}, x_t),
\end{align*}
where
\begin{align*}
	p_{S_t | S_{t-1}, X_t}(s_t | s_{t-1}, x_t)
	&\defeq
	\begin{cases}
		\begin{array}{l}
			\hspace*{-6pt}1
			\mbox{}
		\end{array}
		& \hspace*{-15pt}
		\begin{array}{l}
			\text{(if $s_t = \vx_{t-\nu+1}^t, s_{t-1} = \vx_{t-\nu}^{t-1})$}
		\end{array} \\\\
		\begin{array}{l}
			\hspace*{-6pt}0
		\end{array}
		& \hspace*{-15pt}
		\begin{array}{l}
			\text{(otherwise)}
		\end{array}
	\end{cases}\hspace{-11pt}, \\[3pt]
	p_{Y_t | S_{t-1}, X_t}(y_t | s_{t-1}, x_t)
	&\defeq
	\frac{1}{{2 \pi \sigma^2}}
	\cdot
	\exp
	\left( 
	- \frac{| y_t - u_t|^2}{2\sigma^2}
	\right),
\end{align*}
with $s_{t-1} = \vx_{t-\nu}^{t-1} (\in \setX^\nu)$ and $u_t = \sum_{\ell=0}^{m} g_{\ell} x_{t-\ell}$.

All possible state sequences of an ISI channel (and more generally, of an
FSMC) can be represented by a trellis diagram. Because of the assumed
time invariance, it is sufficient to show a single trellis section. For
example, Fig.~\ref{Fig.trellis}($a$) shows a trellis section of an ISI channel
characterized by the couple $\bigl( g(D) \defeq 1-D, \ \sigma^2 \bigr)$
with $\nu=1$ and input alphabet $\set{X} \defeq \{ +1, -1 \}$. In
this diagram, branches start at state $s_{t-1} \,\defeq\, x_{t-1}$, end at state
$s_t \,\defeq\,\, x_t$, and have noiseless channel output symbol $u_t = x_t - x_{t-1}$
shown next to them.

\begin{figure*}
	\centering
	\includegraphics[scale=0.70]{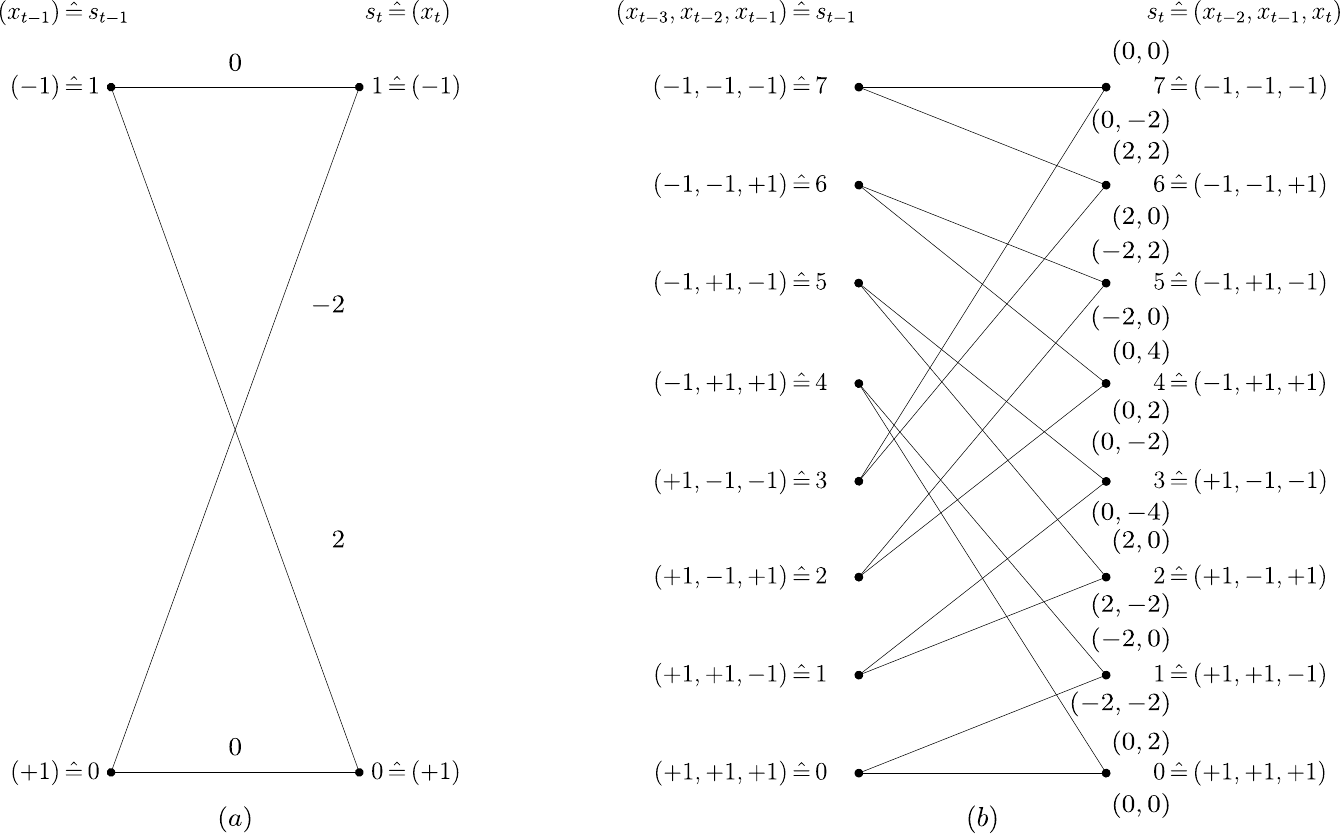}
	\caption{%
		($a$)~Trellis section of an FSMC, modeling an ISI channel with
			$g(D) = 1 - D$, when used with $\nu=1$ and
		$\setX  =  \{ +1, -1 \}$. The noiseless channel output symbol $u_t$ is shown
		next to the branches. %
		($b$)~Trellis section of an FSMC, modeling an ISI-WTC with $\gB(D) = 1 - D$ and
			$\gE(D) = 1 + D - D^2 - D^3$, when used with $\nu=3$ and
		$\setX  =  \{ +1, -1 \}$. Noiseless channel output symbols
		$(u_t,v_t)$, one noiseless channel
		output symbol for Bob's channel and one noiseless channel output symbol for
		Eve's channel, are shown next to the branches.\\[-16pt]  \mbox{}}
	\label{Fig.trellis}
\end{figure*}

Let $\setB$ $(\subseteq \setS \times \setS)$ denote the set of all valid consecutive state pairs
$(s_{t-1}, s_t)$ for which $p_{S_t | S_{t-1}}(s_t | s_{t-1})$ is allowed to be non-zero for any $t \in \Z$.
	Moreover, let
	\begin{align*}
	  \setSright{i} 
	     \defeq
	       \bigl\{ 
	         j 
	       \bigm| 
	         (i,j) \in \setB 
	       \bigr\}, \quad
	  \setSleft{j} 
	     \defeq
	       \bigl\{
	         i 
	       \bigm|
	        (i,j) \in \setB
	       \bigr\},
	\end{align*} 
	 be the set of states $S_t\in\setS$ reachable from $S_{t-1} = i$ and the 
	 set of states $S_{t-1}\in\setS$ that can reach $S_t = j$, respectively.
	For every $(i,j) \in \setB$, let $p_{ij}\defeq p_{S_t|S_{t-1}}(j|i)$ be the
	time-invariant state transition probability assigned by an ergodic and
	non-periodic Markov source of memory order $\nu$.
	Then there is a unique stationary state PMF
	$\{\mu_i\}_{i\in\setS}$ such that $p_{S_t}(i) = \mu_i$ for all $t \in \Z$,
	$i \in \setS$. Finally, let $Q_{ij} \defeq \mu_i \cdot p_{ij}$, for all
	$(i,j) \in \setB$.


In the above statements, we started with $\{ p_{ij} \}_{(i,j) \in \setB}$ and
derived $\{ \mu_i \}_{i \in \setS}$ and $\{ Q_{ij} \}_{(i,j) \in \setB}$ from
$\{ p_{ij} \}_{(i,j) \in \setB}$. However, for analytical purposes, it turns
out to be beneficial to start with $\{ Q_{ij} \}_{(i,j) \in \setB}$ and derive
$\{ p_{ij} \}_{(i,j) \in \setB}$ and $\{ \mu_i \}_{i \in \setS}$ from
$\{ Q_{ij} \}_{(i,j) \in \setB}$. Note that the set of all valid
$\{ Q_{ij} \}_{(i,j) \in \setB}$ for a fixed set $\set{B}$ is given by the
polytope $\setQ(\setB)$, where
\begin{align*}
  \setQ(\setB)
    &\defeq
       \left\{
         \{ Q_{ij} \}_{(i,j) \in \setB} \ 
       \middle|
         \begin{array}{r@{\ }c@{\ }l}
           Q_{ij} 
             &\geq& 
               0, 
                 \ \forall(i,j) \in \setB, \\
           \sum\limits_{(i,j) \in \setB}Q_{ij}
             &=& 
               1, \\
           \sum\limits_{j\in\setSright{i}}
             Q_{ij} 
             &=& 
               \sum\limits_{k \in \setSleft{i}}
                 Q_{ki}, 
                   \ \forall i \in \setS \!
               \end{array}
           \!\!\right\}\!.
\end{align*}
(See~\cite{4494705} for similar observations.) In the following, we will use
the short-hand notation $\matrQ$ for $\{ Q_{ij} \}_{(i,j) \in
	\setB}$. Moreover, similar to \cite[Assumption~34]{4494705}, we will only be
interested in sets $\setB$ where the Markov sources corresponding to
relative interior points of $\setQ(\setB)$ are ergodic and
non-periodic.

\begin{remark}[\emph{Parameterized family of $\matrQ$}]
	\label{remark:Q:parameterized:family:1}
	
	Frequently, we will consider the setup where
	$\matrQ$ is a function of some parameter~$\theta$. More precisely, for every
	$(i,j) \in \setB$, we let $Q_{ij}(\theta)$ be a smooth function of the
	parameter~$\theta$, where $\theta$ varies over a suitable range. For every~$\theta$, we require
		that $\matrQ(\theta) \defeq \bigl\{ Q_{ij}(\theta) \bigr\}_{(i,j) \in \setB} \in
		\setQ(\setB)$. Moreover, for every $(i,j) \in \setB$, we denote the
	derivative of $Q_{ij}(\theta)$ w.r.t. $\theta$ and evaluated at $\ttheta$ by
	$Q^{\theta}_{ij}(\ttheta)$. We denote the corresponding steady-state and
	state transition probabilities parameterized by~$\theta$ by $\mu_i(\theta)$
	and $p_{ij}(\theta)$, respectively. Similarly, we denote their derivatives
	w.r.t. $\theta$ and evaluated at $\ttheta$ by $\mu_i^{\theta}(\ttheta)$ and
	$p_{ij}^{\theta}(\ttheta)$, respectively. Because
	$\matrQ(\theta) \in \setQ(\setB)$, we have
	$\sum_{(i,j)\in\setB} Q^{\theta}_{ij}(\ttheta) = 0,$ and $
	\sum_{i\in{\setS}} \mu_i^{\theta}(\ttheta) = 0$.
	\remarkend
\end{remark}

\begin{definition}[\textit{{Intersymbol Interference Wiretap Channel (ISI-WTC)}}]
	\label{def:prwt:channel:1}
	
	In this paper, we consider an ISI-WTC, where Alice transmits data symbols over Bob's channel and over
	Eve's channel, which are both assumed to be ISI channels with finite input
	alphabet $\set{X} \subsetneq \C$. (See Fig.~\ref{FIG:BLK_diag}.) Specifically, Bob's channel is an ISI channel described by the couple $\bigl( \gB(D), \sigmaB^2 \bigr)$,
	with transfer polynomial $\gB(D) = \sum_{t=0}^{\mB} \gBt D^t$, noiseless
	output process $\{ U_t \}_{t \in \Z}$, noise process
	$\{ \NBidx{t} \}_{t \in \Z}$, and noisy output process
	$\{ Y_t \}_{t \in \Z}$. Similarly, Eve's channel is an ISI channel
	described by the couple $\bigl( \gE(D), \sigmaE^2 \bigr)$, with transfer
	polynomial $\gE(D) = \sum_{t=0}^{\mE} \gEt D^t$, noiseless output process
	$\{ V_t \}_{t \in \Z}$, noise process $\{ \NEidx{t} \}_{t \in \Z}$, and
	noisy output process $\{ Z_t \}_{t \in \Z}$. We assume that the noise
	process of Bob's channel and the noise process of Eve's channel are
	independent. Clearly, the ISI-WTC is parameterized by the quadruple
	$\bigl( \gB(D), \gE(D), \sigmaB^2, \sigmaE^2 \bigr)$.   \defend
\end{definition}


By choosing $\nu\geq\max (\mB,\mE)$, FSMCs and their associated trellises
	can be used for visualizing ISI-WTCs as well. Since such trellis
	representations are well known, we omit the details and conclude this
	section with the following example.
	(See~\cite{Nouri:Asvadi:Chen:Vontobel:21:1} and \cite{5351376} for more
	details.)

Consider an ISI-WTC with $\setX = \{ +1, -1 \}$, where Bob's channel is described by
$\gB(D) = 1 - D$ (see Fig.~\ref{Fig.trellis}($a$)), and where Eve's channel
is described by $\gE(D) = 1 + D - D^2 - D^3$. Let $\nu = 3$.
Then all possible state sequences of an FSMC modeling this ISI-WTC can be represented by a trellis
diagram. Because of the assumed time invariance, it is sufficient to show a
single trellis section, as shown in Fig.~\ref{Fig.trellis}($b$) for the
present example.


\subsection{Secure Rate}
\label{sec::Pre.Ach}

\begin{definition}\label{def:code}
	An $\bigl( {e^{n\Rs}}, n \bigr)$ code for the wiretap channel consists
		of a message set $\setM$ with $|\setM| \defeq \bceil{e^{n\Rs}}$, a
		stochastic encoder $f:\setM\to\setX^n$, and a decoder
		$\phi:\C^n\to\setM$.\footnote{The message set $\setM$, the encoding
				function $f$, and the decoding function $\phi$ implicitly depend on
				the block length~$n$.} \defend 
\end{definition}

Let $\Mn$ be a random variable corresponding to a uniformly chosen secret message from the alphabet $\setMn$. 
	The reliability of Bob’s decoder is measured by the probability of a block error $\Pr\big(M\neq\phi(\vY^n)\big)$ and the secrecy performance of the code is measured by the statistical independence between $\Mn$ and $\vZ_1^n$ in terms of the variational distance
	$d_{\setMn \times \setZ^n}(p_{\Mn, \vZ_1^n}, p_{\Mn} p_{\vZ_1^n})$. (See Appendix~\ref{apndx.frst} for more details.)

\begin{definition}\label{def:achievability}
	A secure rate $\Rs$ is said to be achievable if there exists a sequence
		of codes~$\bigl( {e^{n\Rs}}, n \bigr)$ as in Definition~\ref{def:code},
		with $n \to \infty$, satisfying the reliability criterion
		\begin{align*}
			\Pr\big(M\neq\phi(\vY^n)\big)
			\rightarrow 0, \alabel{equ:relcrit}
		\end{align*}
		and the secrecy criterion
		\begin{align*}	
			d_{\setMn \times \setZ^n}(p_{\Mn, \vZ_1^n}, p_{\Mn} p_{\vZ_1^n})
			\rightarrow 0. \alabel{equ:seccrit}
		\end{align*}
		The secrecy capacity is the supremum of all achievable secure rates.  \defend
\end{definition}

\begin{definition}\label{def:SecRate}
	Consider an ISI-WTC with an input Markov source
		described by $\matrQ \in \setQ(\setB)$.	For all $(i,j) \in \setB$,
		define $\TB_{ij}(\matrQ)$ and $\TE_{ij}(\matrQ)$ as	shown in~\eqref{eq:def:TB:1}
		and~\eqref{eq:def:TE:1} at the bottom of this page.\footnote{The expressions $\cTB_{ij}(\matrQ,\vy_1^n)$ and
		$\cTE_{ij}(\matrQ,\vz_1^n)$ in~\eqref{eq:def:TB:1}
		and~\eqref{eq:def:TE:1} are similar to the expression for $\check{T}_{ij}^{(N)}$
		in~\cite[Lemma~70]{4494705}, part ``second possibility.''}
	\defend
\end{definition}

\bigformulabottom{2}{5}{%
	\begin{alignat}{3}
		\!\!\TB_{ij}(\matrQ)
		&\defeq
		\lim_{n\to\infty} \hspace{-0.05cm}
		\int
		p_{\vY_1^n}(\vy_1^n)
		\cdot
		\cTB_{ij}(\matrQ,\vy_1^n)
		\ \mathrm{d}{\vy_1^n},
		\ \
		&&\cTB_{ij}(\matrQ,\vy_1^n)
		&&\defeq
		\frac{1}{n}
		\sum_{t=1}^n
		\log
		\left(
		\frac{p_{S_{t-1}, S_t | \vY_1^n}(i, j | \vy_1^n)
			^{{p_{S_{t-1}, S_t | \vY_1^n}
					(i, j | \vy_1^n)}
				/ {\mu_i p_{ij}}}}
		{p_{S_{t-1} | \vY_1^n}
			(i | \vy_1^n)
			^{{p_{S_{t-1} | \vY_1^n}
					(i | \vy_1^n)}
				/ {\mu_i}}}
		\right)\!,\!
		\alabel{eq:def:TB:1} \\
		\!\!\TE_{ij}(\matrQ)
		&\defeq
		\lim_{n\to\infty} \hspace{-0.05cm}
		\int
		p_{\vZ_1^n}(\vz_1^n)
		\cdot
		\cTE_{ij}(\matrQ,\vz_1^n)
		\ \mathrm{d}{\vz_1^n},
		\ \ 
		&&\cTE_{ij}(\matrQ,\vz_1^n)
		&&\defeq
		\frac{1}{n}\sum_{t=1}^n
		\log
		\left(
		\frac{p_{S_{t-1}, S_t | \vZ_1^n}
			(i, j | \vz_1^n)
			^{{p_{S_{t-1}, S_t | \vZ_1^n}
					(i, j| \vz_1^n)}
				/ {\mu_i p_{ij}}}}
		{p_{S_{t-1} | \vZ_1^n}
			(i|\vz_1^n)
			^{{p_{S_{t-1} | \vZ_1^n}
					(i | \vz_1^n)}
				/ {\mu_i}}}
		\right)\!,\!
		\alabel{eq:def:TE:1} \\
		\bpsi_{\tmatrQ}(\matrQ)
		&\defeq
		\kappa'
		\cdot 
		\Bigg(
		\sum_{(i,j) \in \setB} 
		\tQ_{ij}
		\cdot
		\bigl(
		1 + \kappa \cdot (\deltaQ)_{ij} 
		\bigr)
		\ \cdot
		&&\hspace*{-2pt}\log
		\bigl(
		1 + &&\hspace*{-12pt}\kappa \cdot (\deltaQ)_{ij}
		\big)
		-
		\sum_{i \in \setS}
		\tmu_i
		\cdot
		\bigl( 1 + \kappa \cdot(\deltamu)_i\bigr)
		\cdot
		\log
		\bigl( 1+\kappa\cdot(\deltamu)_i \bigl)
		\Bigg). \alabel{eq:def:bpsi:1}
	\end{alignat}
}

\begin{proposition}\label{Prop:achiev}
	Consider an ISI-WTC with an input Markov source
		described by $\matrQ \in \setQ(\setB)$. Let
		\begin{equation}
			\label{equ.concap}
			\Rs(\matrQ)
				\defeq
				\bigg( \sum_{(i,j)\in\setB} Q_{ij} \cdot \bigl(
				\TB_{ij}(\matrQ) - \TE_{ij}(\matrQ) \bigr) \bigg)^{\!\! +}.
		\end{equation}
		Then all secure rates $\Rs$ satisfying
		\begin{align*}
			R_s < \Rs(\matrQ) 
		\end{align*}
		are achievable over this ISI-WTC under the reliability criterion~\eqref{equ:relcrit} and the secrecy criterion~\eqref{equ:seccrit}.
\end{proposition}

\begin{proof}
	See Appendix~\ref{apndx.scnd}.
\end{proof}

The expressions in~\eqref{eq:def:TB:1}
and \eqref{eq:def:TE:1} make it appear very unlikely that there is a closed-form expression
for $\Rs(\matrQ)$ in terms of $\matrQ$. Accordingly, the best one can do is to estimate
$\Rs(\matrQ)$ for a specific $\matrQ \in \setQ(\setB)$ through Monte-Carlo methods (e.g., variants of
the algorithms in~\cite{1661831}). 

We are now in a position to introduce the notion of constrained secrecy
capacity, which is a key quantity to be studied in the subsequent parts of
this paper.

\begin{definition}
	\label{def:constrained:capacity:1}
	
	Consider an ISI-WTC with an input Markov source
	described by $\matrQ$, varying over $\setQ(\setB)$. The constrained secrecy
	capacity (or, more precisely, the $\setQ(\setB)$-constrained secrecy
	capacity) is defined as
	\begin{align*}
	C_{\setQ(\setB)} \defeq \max_{\matrQ \in
		\setQ(\setB)} \Rs(\matrQ).
	\end{align*}
	\defend
\end{definition}

Roughly speaking, $C_{\setQ(\setB)}$ is the tightest lower bound on the secrecy
capacity of an ISI-WTC that can be obtained by optimizing an input Markov
source described by $\matrQ \in \setQ(\setB)$. 

The next section (Section~\ref{sec::optim}) discusses an efficient
algorithm for finding a local maximum of $\Rs(\matrQ)$ over
$\setQ(\setB)$. In a first reading of this paper, readers might want to skip
this (rather technical) section and go straight to Section~\ref{sec::simul}
that discusses some simulation results.


\section{Secure Rate Optimization}
\label{sec::optim}


A first challenge when optimizing the function $\Rs(\matrQ)$ over
	$\setQ(\setB)$ is the fact that we do not have a closed-form expression for
	$\Rs(\matrQ)$, i.e., the best we can do is to approximate $\Rs(\matrQ)$ by
	estimating it with the help of Monte-Carlo methods. However, instead of
	applying a standard zeroth-order optimization method that is based on
	estimates of $\Rs(\matrQ)$, in this section we pursue a more efficient
	approach that is based on estimates of the gradient of $\Rs(\matrQ)$.%
	\footnote{Note that in the following we ignore the, essentially irrelevant,
	$(\ldots)^+$ operator in \eqref{equ.concap} when approximating $\Rs(\matrQ)$
	and when optimizing $\Rs(\matrQ)$ over $\matrQ\in\setQ(\setB)$.}

A second challenge when optimizing the function $\Rs(\matrQ)$ over
	$\setQ(\setB)$ is the fact that $\Rs(\matrQ)$ is (typically, according to
	our numerical investigations in Section~\ref{sec::simul}) a fluctuating,
	non-concave function. Therefore, we will aim at finding a local maximum of
	$\Rs(\matrQ)$ instead of the global maximum. (Of course, by running the
	optimization algorithm with different initializations, one can potentially
	get different local maxima. Finally, the maximum among all local maxima is
	then selected. See also the discussion in Section~\ref{sec::simul}.)

Our iterative optimization method operates as follows:
\begin{itemize}
	
	\item Assume that at the current iteration the algorithm has found a Markov source
	described by $\tmatrQ \defeq \bigl\{ \tQ_{ij} \bigr\}_{(i,j) \in \setB}$.
	
	\item Around $\matrQ = \tmatrQ$, the algorithm approximates the secure rate function $\Rs(\matrQ)$
	by the surrogate function $\psi_{\tmatrQ}(\matrQ)$
	over $\setQ(\setB)$ satisfying the following properties:
	\begin{itemize}
		
		\item The value of $\psi_{\tmatrQ}(\matrQ)$ matches the value of
		$\Rs(\matrQ)$ at $\matrQ = \tmatrQ$.
		
		\item The gradient of $\psi_{\tmatrQ}(\matrQ)$ w.r.t.\ $\matrQ$ matches
		the gradient of $\Rs(\matrQ)$ w.r.t.\ $\matrQ$ at $\matrQ = \tmatrQ$.
		
		\item The function $\psi_{\tmatrQ}(\matrQ)$ is concave in terms of
		$\matrQ$ and can be efficiently maximized.
		
	\end{itemize}
	
	\item Replace $\tmatrQ$ with the $\matrQ$ maximizing $\psi_{\tmatrQ}(\matrQ)$.
	
\end{itemize}
As sketched in Fig.~\ref{Fig.OptimizationProc}, a well-defined concave
surrogate function with the mentioned properties enables us to search
throughout the polytope $\setQ(\setB)$ and find a local maximum of
$\Rs(\matrQ)$ over $\setQ(\setB)$, iteratively.  Similar techniques
have also been proposed in~\cite{4494705, 4777638}.


\begin{figure}
	\centering
	\includegraphics[scale=0.68]{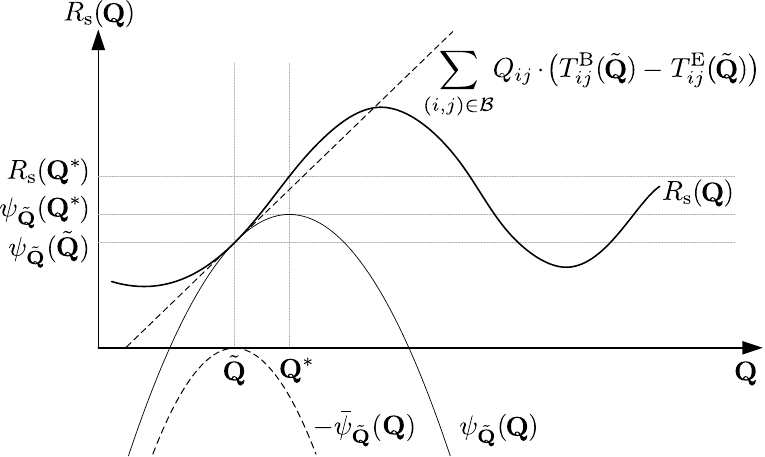}
	\caption{Sketch of the functions appearing in the optimization algorithm
		discussed in Section~\ref{sec::optim}. Note that while the domain is
		one-dimensional in this sketch, it is $|\setB|$-dimensional in the actual
		optimization problem. \\ \mbox{}}
	\label{Fig.OptimizationProc}
\end{figure}


\subsection{The Surrogate Function}


In this section, we first introduce the surrogate function. Then, we show that the
employed surrogate function fulfills the promised properties.

In the following, in the same way that we derived
$\{ p_{ij} \}_{(i,j) \in \setB}$ and $\{ \mu_i \}_{i \in \setS}$ from
$\matrQ = \{ Q_{ij} \}_{(i,j) \in \setB}$\,, we will derive
$\{ \tp_{ij} \}_{(i,j) \in \setB}$ and $\{ \tmu_i \}_{i \in \setS}$ from
$\tmatrQ = \{ \tQ_{ij} \}_{(i,j) \in \setB}$.
For every $(i,j) \in \setB$, let
\begin{align*}
	(\deltaQ)_{ij}
	\defeq
	\frac{Q_{ij} - \tQ_{ij}}
	{\tQ_{ij}},
	\
	(\deltamu)_i
	\defeq
	\frac{\mu_i - \tmu_i}
	{\tmu_i}.
\end{align*}
Moreover, let $\bpsi_{\tmatrQ}(\matrQ)$ be as defined in \eqref{eq:def:bpsi:1}, shown at the bottom of the previous page, where the real parameters $0 < \kappa \leq 1$ and $\kappa' > 0$ are used to
control the shape of $\bpsi_{\tmatrQ}(\matrQ)$. For a given operating point
$\tmatrQ\in\setQ(\setB)$, the surrogate function is specified by (see also Fig.~\ref{Fig.OptimizationProc})
\begin{align}
	\label{equ.Surrogate}
	\psi_{\tmatrQ}(\matrQ)
	&\defeq
	\underbrace{
		\sum_{(i,j) \in \setB}
		Q_{ij}
		\cdot
		\bigl(
		\TB_{ij}(\tmatrQ)
		-
		\TE_{ij}(\tmatrQ)
		\bigr)
	}_{\text{\ding{172}}}
	\underbrace{
		-\
		{\bpsi_{\tmatrQ}(\matrQ)}_{\parbox{0.1cm}{\mbox{} \\[-0.075cm]}}\!.
	}_{\text{\ding{173}}}
\end{align}
Note that, for a fixed $\tmatrQ$, the expression~\ding{172} is linear in
$\matrQ$, while the expression~\ding{173} is concave in $\matrQ$ and has a
zero gradient for $\matrQ = \tmatrQ$. The role of the expression~\ding{172} is
to be a first-order approximation of $\Rs(\matrQ)$ at $\matrQ = \tmatrQ$,
while the role of~\ding{173} is to regularize $\psi_{\tmatrQ}(\matrQ)$. While
many other expressions than~\ding{173} could have been chosen as a
regularization term, the expression in~\ding{173} yields the following
desirable features for $\psi_{\tmatrQ}(\matrQ)$: first, the function
$\psi_{\tmatrQ}(\matrQ)$ can be efficiently maximized over $\matrQ$, second,
maximizing $\psi_{\tmatrQ}(\matrQ)$ implicitly also improves the entropy rate
of the input Markov source described by $\matrQ$.\footnote{Let us make the latter statement
	more precise for $\kappa = 1$ and $\kappa' = 1$. Namely, after some
	algebraic manipulations, one obtains
	$- \bpsi_{\tmatrQ}(\matrQ) = - \sum_{(i,j) \in \setB} Q_{ij} \cdot
	\log(p_{ij}) + \sum_{(i,j) \in \setB} Q_{ij} \cdot \log(\tp_{ij})$.
	Here, the first term equals the entropy rate of a Markov source,
	whereas the latter term,  which is linear in $\matrQ$, guarantees a zero gradient of
	$- \bpsi_{\tmatrQ}(\matrQ)$ for $\matrQ = \tmatrQ$.}



In the following, we examine the promised properties of the employed surrogate
function~(\ref{equ.Surrogate}). For brevity, we use the short-hand notations
$\Rs(\theta)$, $\psi_{\tmatrQ}(\theta)$, and $\tmatrQ$ for
$\Rs\bigl( \matrQ(\theta) \bigr)$,
$\psi_{\tmatrQ}\bigl( \matrQ(\theta) \bigr)$, and
$\matrQ(\ttheta)\in\setQ(\setB)$, respectively.


\begin{lemma}[\emph{Property~1 of the surrogate function $\psi$}]
	\label{lem.rCalc}
	
	The value of $\psi_{\tmatrQ}(\matrQ)$ matches the value of
	$\Rs(\matrQ)$ at $\matrQ = \tmatrQ$, i.e.,
	$ 
	\psi_{\tmatrQ}(\tmatrQ)
	= \Rs(\tmatrQ),
	$ 
	and, in terms of the parameterization defined above,
	$ 
	\label{equ.frstorder}
	\psi_{\tmatrQ}(\ttheta)
	= \Rs(\ttheta).
	$ 
\end{lemma}


\begin{proof}
	We start by noting that $\matrQ=\tmatrQ$ implies $(\deltaQ)_{ij}=0$
	and $(\deltamu)_i = 0$ for all $(i,j) \in \setB$, which in turn
	implies that $\bpsi_{\tmatrQ}(\tmatrQ) = 0$. The result $\psi_{\tmatrQ}
	(\tmatrQ) = \Rs(\tmatrQ)$ follows then from \eqref{equ.Surrogate} along
	with \eqref{equ.concap} in Proposition~\ref{Prop:achiev}.
\end{proof}


\begin{lemma}[\emph{Property~2 of the surrogate function $\psi$}]
	\label{lem.nuCalc}
	
	The gradient of $\psi_{\tmatrQ}(\matrQ)$ w.r.t.\ $\matrQ$ matches the
	gradient of $\Rs(\matrQ)$ w.r.t.\ $\matrQ$ at $\matrQ = \tmatrQ$, i.e.,
	 \begin{align*}
	     \left.
	       \rdiff{\theta}
	         \psi_{\tmatrQ}(\theta)
	     \right|_{\theta = \ttheta}
	         = \left.
	            \rdiff{\theta}
	              \Rs(\theta)
	          \right|_{\theta = \ttheta}
	 \end{align*}
	for any parameterization as defined above.
\end{lemma}


\begin{proof}
	We start by showing that
	$
	\left.
	\rdiff{\theta}
	\bpsi_{\tmatrQ}(\theta)
	\right|_{\theta = \ttheta} 
	= 0.
	$
	Indeed,
	\begin{align*}
			\left.
			\rdiff{\theta}
			\bpsi_{\tmatrQ}(\theta)
			\right|_{\theta = \ttheta} \hspace{-0.11cm}
			&= \Biggl.
			\kappa 
			\kappa'
			\cdot
			\Biggl(
			\sum_{(i,j) \in \setB}\hspace{-0.2cm}
			Q^{\theta}_{ij}(\theta)
			\cdot 
			\log
			\bigl( 
			1 + \kappa \cdot (\deltaQ(\theta))_{ij}
			\bigr)\\
			&-
			\sum_{i\in\setS}
			\mu_i^\theta(\theta)
			\cdot
			\log
			\bigl(
			1 + \kappa\cdot(\deltamu(\theta))_i
			\bigr)
			\Biggr)
			\Biggr|_{\theta = \ttheta}
			= 0.\alabel{eq:bpsi:property:2}
	\end{align*}
	We then have
	\begin{align*}
		&\left.
		\rdiff{\theta}
		\psi_{\tmatrQ}(\theta)
		\right|_{\theta = \ttheta}
		= \left.
		\rdiff{\theta}
		\bigl(
		\psi_{\tmatrQ}(\theta)
		+
		\bpsi_{\tmatrQ}(\theta)
		\bigr)
		\right|_{\theta = \ttheta}\\
		&\qquad= \left.
		\rdiff{\theta}
		\left(
		\sum_{(i,j) \in \setB}
		Q_{ij}(\theta)
		\cdot
		\bigl(
		\TB_{ij}(\ttheta)
		-
		\TE_{ij}(\ttheta)
		\bigr)
		\right)
		\right|_{\theta = \ttheta} 
		\nonumber \\
		&\qquad= \left.
		\rdiff{\theta}
		\left(
		\sum_{(i,j)\in\setB}
		Q_{ij}(\theta)
		\cdot
		\bigl(
		\TB_{ij}(\theta)
		-
		\TE_{ij}(\theta)
		\bigr)
		\right)
		\right|_{\theta = \ttheta}\\
		&\qquad= \left.
		\rdiff{\theta}
		\Rs(\theta)
		\right|_{\theta = \ttheta},
		\alabel{eq:bpsi:property:3}
	\end{align*}
	where the first equality follows from~\eqref{eq:bpsi:property:2}, the
	second equality follows from \eqref{equ.Surrogate}, the third
	equality follows from~\cite[Lemma~64]{4494705}, and the fourth
	equality follows from \eqref{equ.concap}. 
\end{proof}


Despite the close similarity between the third and the fourth expressions
in~\eqref{eq:bpsi:property:3}, this is a non-trivial result because of the
non-triviality of~\cite[Lemma~64]{4494705}. 


\begin{lemma}[\emph{Convexity of the function $\bpsi_{\tmatrQ}$}]
	\label{lem.convexity}
	
	The function $\bpsi_{\tmatrQ}(\matrQ)$ is convex over
	$\matrQ \in \setQ(\setB)$.
\end{lemma}


\begin{proof}
	See Appendix~\ref{apndx.convexity}.
\end{proof}


\begin{lemma}[\emph{Property~3 of the surrogate function $\psi$}]
	The surrogate function $\psi_{\tmatrQ}(\matrQ)$ is concave over
	$\matrQ \in \setQ(\setB)$.
\end{lemma}


\begin{proof}
	This follows immediately from Lemma~\ref{lem.convexity} and from
	$\sum_{(i,j)\in\setB}Q_{ij}\cdot\big( \TB_{ij}(\tmatrQ) -
	\TE_{ij}(\tmatrQ) \big)$ being a linear function in terms of $\matrQ$.
\end{proof}


	
\subsection{Maximizing the Surrogate Function}


Let $\tmatrQ \in \setQ(\setB)$ denote the parameter of a Markov source attained at the
current iteration of the proposed algorithm. In the next iteration, $\tmatrQ$
is replaced by $\matrQ^{*} = \bigl\{ Q^{*}_{ij} \bigr\}_{(i,j) \in \setB}$,
where
\begin{align}
	\matrQ^{*}
	&\defeq
	\underset{\matrQ \in \setQ(\setB)}%
	{\operatorname{arg\ max}} \ 
	\psi_{\tmatrQ}(\matrQ).
	\label{eq:surrogate:function:optimization:1}
\end{align}


\begin{proposition}[\emph{The optimum distribution
		$\matrQ^{*}$}]\label{prop.DistEig}
	
	The optimum Markov source distribution $\matrQ^{*}$
	in~\eqref{eq:surrogate:function:optimization:1} is calculated as
	follows. Let $\matrA \defeq \bigl( A_{ij}\bigr)_{i,j \in \setS}$ be the
	matrix with entries
	\begin{align}
		\label{equ.Aij}
		A_{ij}
		&\defeq
		\begin{cases}
			{\displaystyle 
				\tp_{ij} 
				\cdot
				\exp
				\left(
				\frac{\tTB_{ij} - \tTE_{ij}}
				{\kappa \kappa'}
				\right)
			}
			& \text{$\big((i,j) \in \setB\big)$} \\
			0
			& \text{(otherwise)}
		\end{cases},
	\end{align}
	where $\tTB_{ij} \defeq \TB_{ij}(\tmatrQ)$ and
	$\tTE_{ij} \defeq \TE_{ij}(\tmatrQ)$ are defined according to
	Definition~\ref{def:SecRate}. Note that $\matrA$ is a non-negative
	matrix, i.e., a matrix with non-negative entries. Let $\rho$ be the
	Perron--Frobenius eigenvalue of the matrix $\matrA$, with the corresponding
	right eigenvector $\vgamma = (\gamma_j)_{j \in \setS}$.\footnote{Recall that
		the Perron--Frobenius eigenvalue of an irreducible non-negative matrix is
		the eigenvalue with the largest absolute value. One can show that the
		Perron--Frobenius eigenvalue is a positive real number and that the
		corresponding right eigenvector can be multiplied by a suitable scalar
		such that all entries are positive real numbers.} Define
	\begin{align}
		\label{equ.smaeee}
		\hp_{ij}^{*}
		&\defeq
		\frac{A_{ij}}{\rho}
		\cdot
		\frac{\gamma_j}{\gamma_i},
		\quad (i,j) \in \setB.
	\end{align}
	Calculate $\{\hQ^{*}_{ij}\}_{(i,j)\in\setB}$ from
	$\{\hp_{ij}^{*}\}_{(i,j)\in\setB}$ (in the same way that we derived
	$\{ Q_{ij} \}_{(i,j) \in \setB}$ from $\{ p_{ij} \}_{(i,j) \in
		\setB}$). If
	\begin{align}
		\label{equ.kappalwrbound}
		\kappa
		&\geq
		\frac{\tQ_{ij} - \hQ^{*}_{ij}}
		{\tQ_{ij}},
		\quad (i,j) \in \setB,
	\end{align}
	then the parameter $\matrQ^{*}$ is given by solving the following system of linear
	equations in terms of $\bigl\{ Q^{*}_{ij} \bigr\}_{(i,j) \in \setB}$
	\begin{align*}
		\hspace*{0pt}\left\{
		\begin{array}{r@{\ }c@{\ }l}
			\hspace{-6pt}Q^{*}_{ij} 
			- 
			\hp_{ij}^{*}
			\sum\limits_{j'\in\setSright{i}}
			Q_{ij'}^{*}
			-\frac{1-\kappa}
			{\kappa}
			\cdot
			\bigl(
			\tmu_i
			\hp_{ij}^{*}
			-
			\tQ_{ij}
			\bigr)&=&0, \>
			(i,j) \in \setB, \\
			\sum_{r\in\setSleft{i}}
			Q_{ri}^{*}
			-
			\sum_{j\in\setSright{i}}
			Q^{*}_{ij}
			&=& 0,
			\> i \in \setS, \\[10pt]
			\sum_{(i,j)\in\setB}
			Q^{*}_{ij}
			&=& 1.
		\end{array}
		\right.
	\end{align*}
\end{proposition}


\begin{proof}
	See Appendix~\ref{apndx.bee}.
\end{proof}


Note that Proposition~\ref{prop.DistEig} applies Perron--Frobenius theory for
irreducible non-negative matrices. One can verify that $\matrA$ is irreducible
except for uninteresting boundary cases. Note also that increasing the real
parameters $\kappa$ and $\kappa'$ makes the surrogate function to be narrower
and steeper, which reduces the aggressiveness of the searching step size.



\begin{algorithm}[ht!]
	\caption{Secure Rate Optimization}
	\begin{spacing}{1.12}
		\label{alg.EM}
		$r \leftarrow 0$\;
		\While{\emph{\emph{convergence occurs}}}{
			$\tmatrQ \leftarrow \matrQ^{\langle r \rangle}$\;
			Generate a sequence $\cvx_1^n$ based on $\tmatrQ$\;
			Simulate Bob's (Eve's) channel with input $\cvx_1^n$ to \parbox{7.5cm}{~\\[-5pt]\hspace*{-5pt}obtain $\cvy_1^n$ ($\cvz_1^n$) at the output\;}\\
			\For{$(i,j) \in \setB$}{\vspace*{2pt}
				Calculate $\cTB_{ij}(\tilde{\matrQ},\cvy_1^n)$ and $\cTE_{ij}(\tilde{\matrQ},\cvz_1^n)$ according \parbox{7.5cm}{~\\[-5pt]\hspace*{-5pt}to~\eqref{eq:def:TB:1} and~\eqref{eq:def:TE:1}}\;\vspace*{5pt}
				$\displaystyle\check A_{ij}\leftarrow 
				\tp_{ij} 
				\cdot
				\exp
				\left(
				\frac{{\cTB_{ij}(\tilde{\matrQ},\cvy_1^n)
						-
						\cTE_{ij}(\tilde{\matrQ},\cvz_1^n)}}
				{\kappa \kappa'}
				\right)$\;
			}
			$\cRs^{\langle r \rangle}\leftarrow  \sum_{(i,j)\in\setB} \tilde{Q}_{ij} \cdot \bigl(
			\cTB_{ij}(\tilde{\matrQ},\cvy_1^n)
			-
			\cTE_{ij}(\tilde{\matrQ},\cvz_1^n) \bigr)^{\!\! +}$\;
			Find the Perron--Frobenius eigenvalue $\check \rho$ and the \parbox{7.5cm}{~\\[-5pt]\hspace*{-5pt}corresponding right eigenvector $\cgamma$ of~$\bigl( \check A_{ij}\bigr)_{i,j \in \setS}$\;}\\ 
			\For{$(i,j) \in \setB$}{\vspace*{2pt}
				$\displaystyle
				\chp_{ij}^{*}
				\leftarrow
				\frac{\check A_{ij}}
				{\check \rho}
				\cdot
				\frac{\check \gamma_j}
				{\check \gamma_i}$\;
			}
			Calculate $\{\chQ^{*}_{ij}\}_{(i,j)\in\setB}$ from
			$\{\chp_{ij}^{*}\}_{(i,j)\in\setB}$ (as we \parbox{7.5cm}{~\\[-5pt]\hspace*{-5pt}derived
				$\{ Q_{ij} \}_{(i,j) \in \setB}$ from
				$\{ p_{ij} \}_{(i,j) \in \setB}$)\;}\\[2pt]
			\uIf{$ 
				\kappa
				\geq
				({\tQ_{ij} - \chQ^{*}_{ij}})/{\tQ_{ij}},
				$ 
				for all $(i,j) \in \setB$}{
				Calculate $\cmatrQ^*$ by solving the
				following system of \parbox{7.5cm}{~\\[-5pt]\hspace*{-5pt}linear equations in terms of $\bigl\{ \check Q^{*}_{ij}
					\bigr\}_{(i,j) \in \setB}$}\vspace*{-2pt}
				\begin{align*}
					\hspace*{-0.3cm}\left\{
					\begin{array}{r@{\ }c@{\ }l}
						\check Q^{*}_{ij} 
						- 
						\chp_{ij}^{*}
						\sum\limits_{j'\in\setSright{i}}
						\check Q_{ij'}^{*}
						-\frac{1-\kappa}
						{\kappa}
						\cdot
						\bigl(
						\tmu_i
						\chp_{ij}^{*}
						-
						\tQ_{ij}
						\bigr)&=&0, \\[-0.25cm]
						(i,j) &\in& \setB, \\
						\sum_{r\in\setSleft{i}}
						\check Q_{ri}^{*}
						-
						\sum_{j\in\setSright{i}}
						\check Q^{*}_{ij}
						&=& 0,\\ 
						i &\in& \setS, \\
						\sum_{(i,j)\in\setB}
						\check Q^{*}_{ij}
						&=& 1;
					\end{array}
					\right.\alabel{equ:Lsystem}
				\end{align*}\vspace*{-8pt}}
			\ElseIf {$
				\kappa
				<
				\frac{({\tQ_{ij} - \chQ^{*}_{ij}})}{\tQ_{ij}},
				$
				for any $(i,j) \in \setB$}{
				Suitably change $\kappa$ and go to Step~6\;
			}
			$r\leftarrow r+1$\;
			$\matrQ^{\langle r \rangle} \leftarrow \cmatrQ^*$\;
		}	
	\end{spacing}
\end{algorithm}


The proposed optimization procedure is summarized in
Algorithm~\ref{alg.EM}. Note that this optimization procedure can be
considered as a variation of the well-known EM
algorithm~\cite{dempster1977maximum} comprised of two steps: Expectation
(E-step) and Maximization (M-step). Namely, identifying a concave surrogate
function around a local operating point resembles the E-step and maximization
of the surrogate function to achieve a higher secure rate corresponds to the
M-step. Given this, Algorithm~\ref{alg.EM} has a similar convergence behavior
as the EM algorithm~\cite{wu83}. 

A similar manipulation as performed in~\cite[Eqs.~(52), (53)]{4494705} shows
that, indeed, $\psi_{\tmatrQ}(\cmatrQ^*)\geq\psi_{\tmatrQ}(\matrQ)$ for all
$\matrQ\in\setQ(\setB)$. Consequently, at each iteration $r$, we have
$\psi_{\matrQ^{\langle r \rangle}}(\matrQ^{\langle r+1 \rangle})\geq
\psi_{\matrQ^{\langle r \rangle}}(\matrQ^{\langle r \rangle})$, where the
equality
$\psi_{\matrQ^{\langle r \rangle}}(\matrQ^{\langle r+1 \rangle})
=\psi_{\matrQ^{\langle r \rangle}}(\matrQ^{\langle r \rangle})$ occurs at
the stationary point of Algorithm~\ref{alg.EM}. The stationary points of the
algorithm correspond to the critical points (i.e., local maxima, local
minima, and saddle points) of $\Rs(\matrQ)$ over the polytope
$\setQ(\setB)$. Since the local minima and the saddle points are not stable
stationary points of Algorithm~\ref{alg.EM}, the algorithm converges to a
local maximum of $\Rs(\matrQ)$ achievable from a starting point
$\matrQ^{\langle 0 \rangle}\in\setQ(\setB)$. 



The complexity of one iteration of Algorithm~\ref{alg.EM} is
$\mathcal{O}\bigl( n\cdot 2^{\nu+1} + (2^\nu)^3 \bigr)$, where
	$n\cdot2^{\nu+1}$ stems from estimating $\cTB_{ij}(\matrQ,\cvy_1^n)$,
	$\cTE_{ij}(\matrQ,\cvz_1^n)$, and the Perron-Frobenius eigenvalue
	$\check \rho$ through Monte-Carlo simulations~\cite{1397934}, where $n$ is
	the number of trellis sections used for the Monte-Carlo simulation, and
where $(2^\nu)^3$ stems from solving the system of linear equations
in~\eqref{equ:Lsystem}. (Potentially, the sparsity of the system of linear
equations in~\eqref{equ:Lsystem} can be used to reduce the latter complexity
estimate.)


\section{Practical Implications and Simulation Results}
\label{sec::simul}

In this section, we approximate an NB-IoT uplink channel in a
challenging environment by an ISI channel. Then, we
describe practically relevant wiretapping scenarios and study the maximum
achievable secure rates by applying Algorithm~\ref{alg.EM} \text{to the resulting
ISI-WTCs.}


\subsection{NB-IoT Uplink Channel}


\begin{figure*}
	\begin{minipage}[c]{0.52\textwidth}
		\centering
		\includegraphics[scale=0.78]{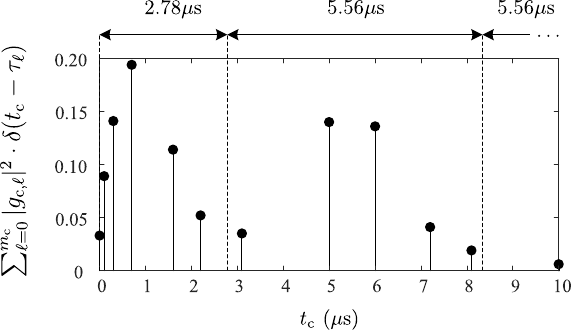}
		\caption{Power-delay profile of the multipath channel. \\ \mbox{}}
		\label{fig:scat}
	\end{minipage}
	\hfill
	\begin{minipage}[c]{0.44\textwidth}
		\captionsetup{justification=centering}
		
		\captionof{table}{\mbox{} \\ 
			\textsc{ISI Channel Model Corresponding to the 
				Power-Delay Profile of~Fig.~\ref{fig:scat},
				with
				$W_{\mathrm{PRB}}^{-1}=5.56~\mu\mathrm{s}$, 
				$f_{\mathrm{carr}}=900$~MHz}}
		\label{tab:simscat}
		\centering
		{\scriptsize
			\begin{tabular}{crcc}
				\hline \\[-5pt]
				$\displaystyle\ell$ & Period ($\displaystyle\mu$s) 	   
				& $\displaystyle|{g}_\ell|^2$ 
				& $\displaystyle e^{-i 2\pi f_\mathrm{carr}\tau_\ell}$ \\
				\hline\\[-5pt]
				0 	& $0-\>\>2.78$ & $0.624$ & 1 \\[5pt] 
				1 	& $2.78-\>\>8.33$ 		   & $0.370$ & 1 \\[5pt] 
				2 	& $8.33-13.89$   & $0.006$ & 1 \\[5pt] 
				\hline
			\end{tabular}
		}
	\end{minipage}
	\vspace*{-0.55cm}
\end{figure*}

Let ${X}(\tc), {Y}(\tc),$ and ${N}(\tc)$ be continuous-time random signals
corresponding to, respectively, the channel's input, the channel's output, and
additive noise.\footnote{The variable $\tc \in \R$ will be used
	to denote continuous time, in order to distinguish it from the discrete time
	variable $t \in \Z$ that is used elsewhere.} The general model for the multipath channel, consisting of a
direct path and $m_{\mathrm{c}}\in\Z$ tapped-delay paths, is described by
$$Y(\tc) \defeq \sum_{\ell=0}^{m_{\mathrm{c}}} |g_{\mathrm{c},\ell}| e^{i \theta_\ell} \cdot
X(\tc-\tau_\ell) + N(\tc),$$ where $i$ denotes the imaginary unit and where the
real parameters $|g_{\mathrm{c},\ell}|$, $\theta_\ell$, and $\tau_\ell$ are the gain, the
phase rotation, and the delay introduced by the $\ell$-th path,
respectively.\footnote{We assume that the local oscillators at the transmitter and
	the receiver terminals are synchronized, so the phase reference
	$\theta_0$ is known. Then, for $\theta_0=0$ and $\tau_0=0$, the phase
	rotation in the $\ell$-th path (w.r.t. the direct path) is given by
	$\theta_\ell=-2\pi f_\mathrm{carr}\tau_\ell$, where $f_\mathrm{carr}$ is the
	carrier frequency.} Fig.~\ref{fig:scat} illustrates a typical power-delay
profile of a multipath channel that was measured in an urban area with
moderate to high tree density~\cite[Fig.\ 2.51]{93141079}.


The uplink channel of the NB-IoT occupies a single physical resource block
(PRB) from the LTE configuration, and so the bandwidth of the transmitted
signal is restricted to
$W_{\mathrm{PRB}}=180$~kHz~\cite[Sec. 5.2.3]{3GPPTS36211}. As can be seen
from~Fig.~\ref{fig:scat}, the delay spread of the wireless channel exceeds the
duration of a single channel use
($W_{\mathrm{PRB}}^{-1}=5.56\mu\mathrm{s}$).\footnote{The statistical channel
	models in COST 207 are valid for applications having an average bandwidth of
	about $200$~kHz~\cite[Sec.\ 2.5.4.2]{93141079}.} This issue along with the
multipath propagation gives rise to ISI. As shown in~Table~\ref{tab:simscat},
the ISI tap coefficients are captured by sampling at times
$\frac{W_{\mathrm{PRB}}^{-1}}{2}+\ell\cdot W_{\mathrm{PRB}}^{-1}$ for
$\ell\in\{0,1,2\}$. By ignoring the third tap, due to
its relative amplitude being 10dB below the first tap, the sampled output of a
filter matched to the shaping pulse at the receiver gives rise to an ISI
channel described by $\big(g(D) \defeq 0.792 + 0.610D, \sigma^2 \big)$.




Before introducing the wiretapping scenario, we consider a point-to-point (P2P)
setup where the only channel input constraint is an average-energy
constraint. This simplification allows us to use the well-known “water-pouring” formulas for
analyzing the capacities of Bob's and Eve's P2P channels in the
examined ISI-WTC.
Let us consider the average-energy constraint per input
symbol $\Es$ (in Joules), the symbol duration $T$ (in seconds), a perfect
lowpass filter of bandwidth $W \defeq \frac{1}{2T}$ with the sampling at
Nyquist frequency $1/T$ at the receiver, and the power spectral density $N(f)$
(in Watts per Hertz) of the additive Gaussian noise before the lowpass filter.
The unconstrained (besides some average-energy constraint) capacity of an ISI
channel, described by $\big(g(D)=\sum_{t=0}^{m} g_t D^t,N(f)\big)$, is given
by the ``water-pouring'' formula (see, e.g.,\cite{We-2003})
\begin{equation*}
	C(g,W) = \frac{1}{2} \cdot \int\limits_{-\infty}^{\infty} \log^+ \left(
	\frac{\alpha} {N(f)/|G(f)|^2} \right) \, \mathrm{d} f,
\end{equation*}
where
 	\begin{align*}
 		G(f)
 		&= \begin{cases}
 			{\displaystyle 
 				\frac{\sum_{\ell=0}^{m} g_{\ell} e^{-i 2 \ell \pi f T}}
 				{\sqrt{\sum_{\ell=0}^{m} |g_{\ell}|^2}}}
 			& \text{(if $|f| \leq W$)} \\
 			0 
 			& \text{(otherwise)}
 		\end{cases},
 	\end{align*}
and where $\alpha > 0$ is chosen such that
\begin{align*}
	\Es = \int\limits_{-\infty}^{\infty}
	{\left(
		\alpha  -  \frac{N(f)}{|G(f)|^2}
		\right)^{\!\!+}}
	\mathrm{d} f.
\end{align*}


\subsection{Wiretapping Scenarios and Achievable Secure Rates}



\begin{figure*}
	\centering
	\includegraphics[scale=1.1]{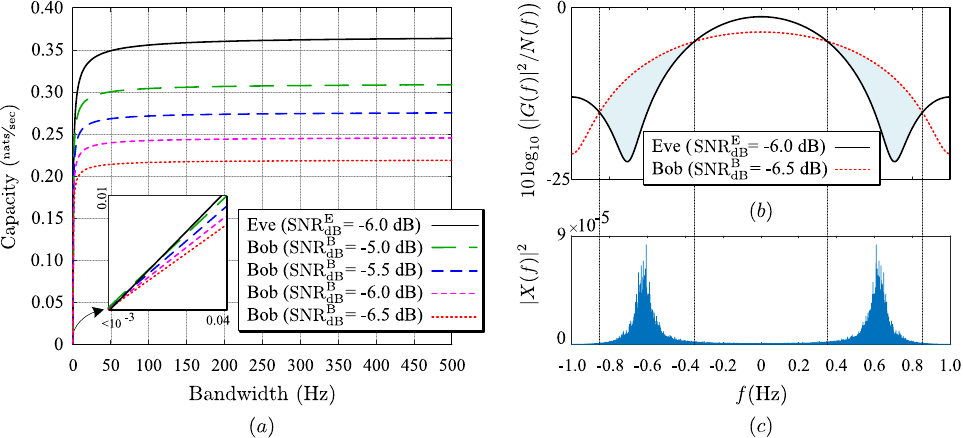}
	\caption{Results for Example~\ref{ex.scnd}, where $\gB(D)= 0.792+ 0.610D$,
		$\gE(D) = 0.446+0.633D+0.633D^2$, $\SNRdBE = -6.0~\text{dB}$. ($a$)~Unconstrained capacities of Bob's and Eve's P2P channels in
		$\mathrm{nats} / \mathrm{sec}$ with normalized average-energy constraint
		$\Es = 1~\mathrm{J}$. ($b$)~Gain-to-noise power spectrum ratios of Bob's
		and Eve's P2P channels in dB/Hz. ($c$)~The power spectrum
		of a sequence generated by the optimized input Markov source. 
		\mbox{}}
	\label{Fig.first.scenario}
\end{figure*}

\begin{figure*}
	\centering \includegraphics[scale=1.1]{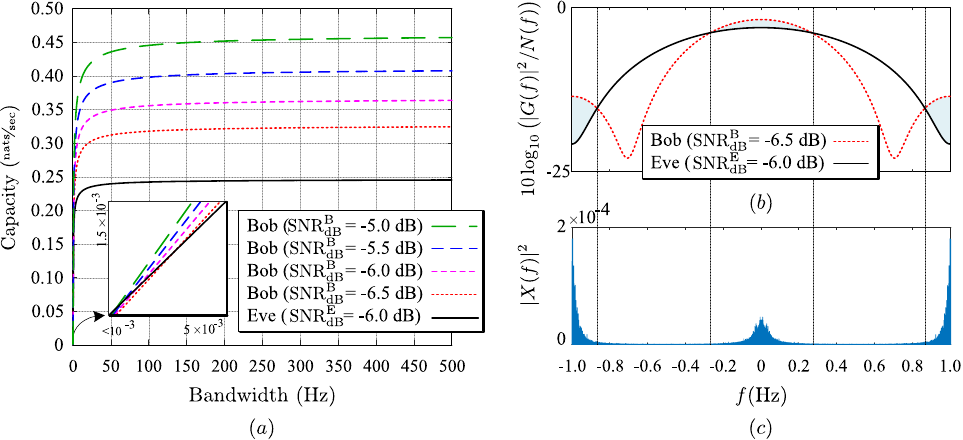}
	\caption{Results for Example~\ref{ex.thrd}, where
		$\gB(D) = 0.446+0.633D+0.633D^2$, $\gE(D) = 0.792 + 0.610D$,
		$\SNRdBE = -6.0~\text{dB}$. ($a$)~Unconstrained capacities of Bob's and
		Eve's P2P channels in $\mathrm{nats} / \mathrm{sec}$ with
		normalized average-energy constraint $\Es = 1~\mathrm{J}$. ($b$)~Gain-to-noise power spectrum ratios of Bob's and Eve's P2P
		channels in dB/Hz. ($c$)~The power spectrum of a sequence generated
		by the optimized input Markov source.}
	\label{Fig.second.scenario}
\end{figure*}


We examine Algorithm~\ref{alg.EM} for optimizing
the parameters of an input Markov source with the alphabet $\setX = \{ +\sqrt{\Es},
	-\sqrt{\Es} \}$ and the memory order $\nu=2$ at the input of two
different ISI-WTCs.\footnote{The BPSK modulation is proposed for the
	narrowband physical uplink shared channel (NPUSCH), both for data (NPUSCH
	Format~1) and control (NPUSCH Format~2)
	channels~\cite[Tab. 10.1.3.2-1]{3GPPTS36211}.} We consider a setup where
Bob's channel and Eve's channel have normalized transfer
polynomials\footnote{A normalized transfer polynomial
	$g(D) \defeq \sum_{t=0}^{m} g_t D^t \in \C[D]$ has to satisfy
	$\sum_{t=0}^{m} |g_t|^2 = 1$. (See, e.g.,~\cite{We-2003}.)} $\gB(D)$ and
$\gE(D)$, and additive white Gaussian noises of variances $\sigmaB^2$ and
$\sigmaE^2$, respectively. Accordingly, the signal-to-noise ratios~(SNRs) of
Bob's channel and Eve's channel are defined as, respectively,
$\SNRB \defeq \Es / \sigmaB^2$ and $\SNRE \defeq \Es / \sigmaE^2$.\footnote{If
	desired, these SNR values can be re-expressed in terms of $\Es/N_0$ values,
	where $N_0/2$ is the two-sided power spectral density of the AWGN process:
	$\Es/N_0 = \frac{1}{2} \cdot (\Es / \sigma^2)$.}

\begin{example}
	\label{ex.scnd}
	
	In the first scenario, Bob's channel is assumed to be the ISI channel
	derived from Table~\ref{tab:simscat}, i.e., $$\gB(D)= 0.792+ 0.610D.$$ Also,
	Eve's channel is assumed to be another ISI channel with the same delay
	profile as in Table~\ref{tab:simscat}, but with different tap
	coefficients. Since it is challenging for Eve to intercept the transmitted
	signals from the line-of-sight transmission~\cite{Ma-2018}, the relative amplitude of Eve's
	direct path is assumed to be (at least) $2.5$~dB below Bob's direct
	path. However, the other tap coefficients are then assumed to be such that
	Eve's channel has the highest unconstrained capacity among all ISI channels
	satisfying the delay profile of Table~\ref{tab:simscat}, i.e.,
	$$(\gE_t)_{t=0}^2=\arg\max_{\tgE: \ |\tgE_0|\leq |\gB_0|-2.5\mathrm{dB}}
	C(\tgE,W).$$ Solving this problem for $\gB_0=0.792$ leads to
	$$\gE(D)=0.446+0.633D+0.633D^2.$$
	
	The resulting unconstrained capacities of
	Bob's and Eve's P2P channels are depicted in
	Fig.~\ref{Fig.first.scenario}($a$).\footnote{Since the NB-IoT protocol
		promises to provide reliable connections with low power consumption, we
		consider low-SNR regimes both for Bob's channel and Eve's
		channel~\cite{8922625}.} It can be seen from
	Fig.~\ref{Fig.first.scenario}($a$) that Eve's channel has a higher
	unconstrained capacity than Bob's channel for sufficiently large enough
	bandwidth. In this sense, Bob's channel is ``worse'' than Eve's
	channel. However, luckily for Bob, there are frequencies where Bob's channel
	has a better gain-to-noise power spectrum ratio than Eve's channel, as can
	be seen from Fig.~\ref{Fig.first.scenario}($b$).  These spectral
	discrepancies can be exploited by a suitably tuned input source toward
	obtaining positive secure rates. Fig.~\ref{Fig.first.scenario}($c$)
		shows the power spectrum of a sequence with the length of $10^6$ generated by the optimized
			Markov source, where the optimization was done with the help of
		Algorithm~\ref{alg.EM}.
		It can be seen from Fig.~\ref{Fig.first.scenario}($c$)
		that the optimized Markov source concentrates the available power of the generated input sequence
		in frequency ranges where Bob's channel has a higher gain-to-noise power
		spectrum ratio than Eve's channel.
	
	Fig.~\ref{Fig.first.CoSecCap} shows the obtained secure rates: on the
	one hand for an \emph{unoptimized} Markov source, producing
	independent and uniformly distributed (i.u.d.)\ symbols, and, on the other
	hand, for an \emph{optimized} Markov source. In this plot, the best obtained secure rate
	is plotted after running Algorithm~\ref{alg.EM} for $100$ different
	initializations.\footnote{The parameters $\kappa$ and $\kappa'$ in
		Algorithm~\ref{alg.EM} took values in the ranges $0.9\leq \kappa \leq 1.0$
		and $4\leq \kappa' \leq 6$, respectively. The initializations
		were generated with the help of Weyl's $|\setS|$-dimensional
		equi-distributed sequences~\cite{RePEc:mtp:titles:0262100711}. (Simulation
		files are available online~\cite{ISIWTC_SF}.)}  \exampleend
\end{example}





In Example~\ref{ex.thrd}, we consider the same
scenario as in Example~\ref{ex.scnd}, but where Bob's channel is swapped with
Eve's channel. For comparison, note that in a \emph{memoryless} wiretap channel
setup, if the first scenario is such that positive secure rates are possible,
then in the second scenario, i.e., after swapping Bob's channel with Eve's channel,
the secure rate is zero~\cite{LeungYanCheong:Hellman:78:1}.

\begin{figure*}[t!]
	\begin{minipage}[b]{.48\textwidth}
		\centering
		\includegraphics[width=1\textwidth]{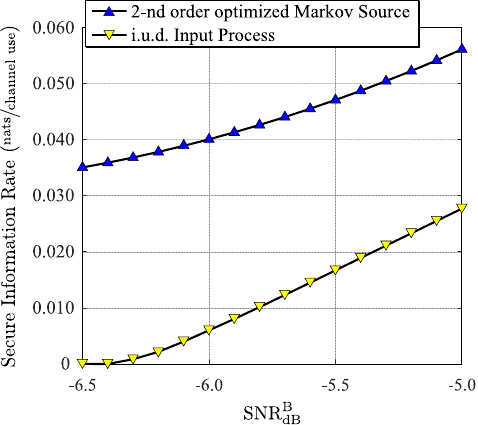}
		\caption{Example~\ref{ex.scnd}:
			Secure rates
			achieved by various input processes in nats/channel use. \\ \mbox{}}
		\label{Fig.first.CoSecCap}
	\end{minipage}
	\hfill
	\begin{minipage}[b]{.48\textwidth}
		\centering
		\includegraphics[width=1\textwidth]{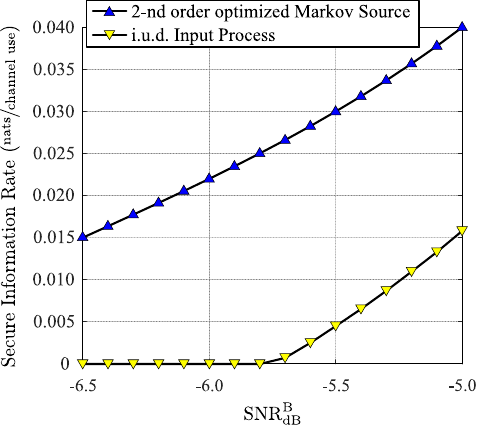}
		\caption{Example~\ref{ex.thrd}:
			Secure rates
			achieved by various input processes in nats/channel use. \\ \mbox{}}
		\label{Fig.second.CoSecCap}
	\end{minipage}
\end{figure*}

\begin{example}
	\label{ex.thrd} 
	
	In the second scenario, the roles of the receiver terminals in
	Example~\ref{ex.scnd} are swapped, i.e.,
	\begin{align*}
		\gB(D)&=0.446+0.633D+0.633D^2,\\
		\gE(D)&=0.792+ 0.610D.
	\end{align*}
	In this case, Bob's channel has a higher
	unconstrained capacity than Eve's channel for large enough
	bandwidth (see Fig.~\ref{Fig.second.scenario}). In this sense, it is not
	unexpected that positive secure rates are possible. Nevertheless, it is
	worthwhile to point out that here positive secure rates are possible even
	though Bob's channel has larger memory than Eve's channel, and for some
	selections of $\SNRdBB$, higher noise power than Eve's channel (see
	Fig.~\ref{Fig.second.CoSecCap}).  \exampleend
\end{example}


\subsection{Discussion}
In a memoryless wiretap channel setup, Eve's channel necessarily has to
be noisier than Bob's channel to achieve a positive secrecy
capacity~\cite{LeungYanCheong:Hellman:78:1}. This results in the capacity of
Eve's channel being less than the capacity of Bob's channel. Interestingly
enough, the optimized Markov sources achieved positive secure rates over
the ISI-WTCs, (i) even when the unconstrained capacity of Bob's channel is
smaller than the unconstrained capacity of Eve's channel (as pointed out in
Example~\ref{ex.scnd}), (ii) even when Bob's channel tolerates both a higher
noise power and a larger memory compared with Eve's channel (as pointed out in
Example~\ref{ex.thrd}). These results confirm the feasibility of
	optimizing input Markov sources for shaping the available power of the
	generated sequences toward benefiting from the spectral discrepancies of
	Bob's and Eve's P2P channels---without consuming any extra power
for cooperative jamming or injecting artificial noise (as it was done
in~\cite{8464866},~\cite{8737786}).


\section{Conclusion}
\label{sec::conclu}


In this paper, we have derived a lower bound on the achievable secure rates
over ISI-WTCs. Then, we have optimized a Markov source at the input of an ISI-WTC
toward (locally) maximizing the obtained secure rates. Because directly
maximizing the secure rate function is challenging, we have iteratively
approximated the secure rate function by concave surrogate functions whose
maximum can be found efficiently. Our numerical results show that by
implicitly using the discrepancies between the frequency responses of Bob's
channel and Eve's channel, it is possible to achieve positive secure rates
also for setups where the unconstrained capacity of Eve's channel is larger
than the unconstrained capacity of Bob's channel.


\appendices


\section{Secrecy Criterion}
\label{apndx.frst}

This appendix gives a concise discussion about the employed secrecy criterion.
	Let $\Mn$ be a random variable corresponding to a uniformly chosen secret
	message from an alphabet $\setMn$. (Note that $\setMn$ implicitly depends on
	the block length $n$.) Moreover, recall that the sequence observed by Eve is
	denoted by $\vZ_1^n$ (see Fig.~\ref{FIG:BLK_diag}). The statistical dependence
	between $\Mn$ and $\vZ_1^n$ is often measured in terms of the mutual
	information between $\Mn$ and $\vZ_1^n$ to ensure the information-theoretic
	perfect secrecy. For instance, the so-called strong secrecy
	criterion~\cite{Maurer1994} requires $I(\Mn; \vZ_1^n) \to 0$ and the so-called
	weak secrecy criterion~\cite{6772207} requires
	$\frac{1}{n} I(\Mn; \vZ_1^n) \to 0$ as $n\to\infty$.
	(See also the recent survey~\cite{9380147}.)

On one hand, the weak secrecy criterion is easier to achieve, but it might
	lead to coding schemes that are vulnerable for practical
	purposes~\cite[Ch. 3.3]{Bloch:2011:PSI:2829193}. On the other hand, the strong
	secrecy criterion is much more desirable, but very difficult to achieve with
	practical coding schemes~\cite{6612711}. Therefore, in the following, we will
	use a secrecy criterion that is stronger than the weak secrecy criterion, but
	more easily achieved than the strong secrecy
	criterion~\cite[Proposition~1]{6612711}. Namely, we use the secrecy criterion \eqref{equ:seccrit},
	 based on the variational distance
	$d_{\setMn \times \setZ^n}(p_{\Mn, \vZ_1^n}, p_{\Mn} p_{\vZ_1^n})$,
	which was called $\mathbb{S}_2(p_{\Mn, \vZ_1^n}, p_{\Mn} p_{\vZ_1^n})$
	in~\cite{6612711}. This secrecy measure can be bounded as
	\begin{align*}
		&d_{\setMn \times \setZ^n}(p_{\Mn, \vZ_1^n}, p_{\Mn} p_{\vZ_1^n})\\
			&=
			\int_{\vz_1^n\in\setZ^n}
				\sum_{\mn\in\setMn}
				p_{\Mn}(\mn)\cdot
				\Big|
			 p_{\vZ_1^n|\Mn}(\vz_1^n|\mn)
				-
			 p_{\vZ_1^n}(\vz_1^n)
				\Big|\mathrm{d}\vz_1^n\\
		&=\int_{\vz_1^n\in\setZ^n}
		\sum_{\mn\in\setMn}
		p_{\Mn}(\mn)\cdot
		\Big|
		p_{\vZ_1^n|\Mn}(\vz_1^n|\mn)\\
		&\qquad\quad-
		\sum_{\tilde \mn \in\setMn}
		p_{\vZ_1^n,\Mn}(\vz_1^n,\tilde \mn)
		\Big| \ \mathrm{d}\vz_1^n\\
		&\leq
		\sum_{(\mn,\tilde \mn)\in \setMnsq}
		p_{\Mn}(\mn)\cdot p_{\Mn}(\tilde \mn)
		\cdot
		\int_{\vz_1^n\in\setZ^n}
		\Big|
		p_{\vZ_1^n|\Mn}(\vz_1^n|\mn)\\
		&\qquad\quad-
		p_{\vZ_1^n|\Mn}(\vz_1^n|\tilde \mn)
		\Big| \ \mathrm{d}\vz_1^n\\
		&=
		\sum_{(\mn,\tilde \mn)\in \setMnsq}
		p_{\Mn}(\mn)\cdot p_{\Mn}(\tilde \mn)\cdot
		d_{\setZ^n}(p_{\vZ_1^n|\Mn=\mn}, p_{\vZ_1^n|\Mn=\tilde \mn}),
		\alabel{equ:indstngshbl}
	\end{align*}
	where the inequality follows from the triangle inequality.
	It follows from (\ref{equ:indstngshbl}) that satisfying (\ref{equ:seccrit})
	makes $\mn,\tilde \mn\in\setMn$ statistically (almost) indistinguishable at
	Eve's decoder.

For further context, note that the secrecy criterion in (\ref{equ:seccrit}) is
	weaker than the so-called distinguishing secrecy criterion in
	cryptography~\cite{10.1007/978-3-642-32009-5_18}, which requires
	$$\max\limits_{(\mn,\tilde \mn)\in
		\setMnsq}\!\big(d_{\setZ^n}(p_{\vZ_1^n|\Mn=\mn}, p_{\vZ_1^n|\Mn=\tilde
		\mn})\big) \to 0,$$ 
		as $n\to\infty$, and which is equivalent to the so-called semantic secrecy criterion.\footnote{Semantic secrecy
		criterion requires that it is impossible for Eve to estimate any function of $\Mn$
		better than to guess it without considering
		$\vZ_1^n$~\cite{10.1007/978-3-642-32009-5_18}.}
 As a consequence, satisfying
	(\ref{equ:seccrit}) gives rise to a loosened notion of the semantic
	security. This looseness arises from the extra assumption that
		$p_{\Mn}$ is fixed and known,  contrary to the cryptographically
		relevant secrecy criteria.\footnote{Generally, from the
			information-theoretic perspective, we assume that a universal source
			encoder is used to compress the data source before data transmission,
			resulting in a sequence that is arbitrarily close to uniformly
			distributed~\cite{1386546}.}


\section{Proof of Proposition~\ref{Prop:achiev}}
\label{apndx.scnd}

We start by defining the notations that will be used in this appendix.
	The \emph{mutual information density} between the respective
	realizations of random variables $X$ and $Y$ is defined to be
	\begin{align*}
		i(x;y)
		\defeq
		\log
		\left(
		\frac{p_{X,Y}(x,y)}
		{p_X(x) \cdot p_{Y}(y)}
		\right).
	\end{align*}
	Moreover, the \emph{conditional mutual information density} between
	the respective realizations of random variables $X$ and $Y$ given $Z=z$
	is defined to be
	\begin{align*}
		i(x;y|z)
		\defeq
		\log
		\left(
		\frac{p_{X,Y|Z}(x,y|z)}
		{p_{X|Z}(x|z) \cdot p_{Y|Z}(y|z)}
		\right).
	\end{align*}
	Consequently, we have
	\begin{align*}
		I(X;Y) 
		&= \sum_{x,y} 
		p_{X,Y}(x,y) \cdot i(x;y), \\
		I(X;Y|Z) 
		&= \sum_{x,y,z} 
		p_{X,Y,Z}(x,y,z) \cdot i(x;y|z).
	\end{align*}
	Following~\cite{spectral}, the spectral sup/inf-mutual information rates are
	defined to be
	\begin{align*}
		\plimsup{n\to\infty}
		&\frac{1}{n}{i}(\vX_1^n;\vY_1^n)\\
		&\defeq
		\inf \left\{ \alpha: \lim_{n\to\infty} 
		\Prob 
		\left( 
		\frac{1}{n}{i}(\vX_1^n;\vY_1^n) > \alpha \right) = 0 
		\right\},\\
		\pliminf{n\to\infty}
		&\frac{1}{n}{i}(\vX_1^n;\vY_1^n)\\
		&\defeq
		\sup \left\{ \beta: \lim_{n\to\infty} 
		\Prob 
		\left(
		\frac{1}{n}{i}(\vX_1^n;\vY_1^n) < \beta \right) = 0 
		\right\}.
	\end{align*}


According to \cite[Lemma~2]{4797642},
	for an arbitrary wiretap channel $\bigl( \setX$,
	$\{ p_{\vY_1^n, \vZ_1^n | \vX_1^n}(\vy_1^n, \vz_1^n | \vx_1^n)
	\}_{n=1}^\infty$, $\setY$, $\setZ\bigr)$ consisting of an arbitrary input
	alphabet $\setX$, two arbitrary output alphabets $\setY$ and $\setZ$
	corresponding to Bob's and Eve's observations, respectively, and a
	sequence of transition probabilities $\{p_{\vY_1^n, \vZ_1^n | \vX_1^n}
	(\vy_1^n, \vz_1^n | \vx_1^n) \}_{n=1}^\infty$, all secure rates $\Rs$
	satisfying
	\begin{align*}
		\Rs
		&\!\! < \!\!\!
		\max_{\{\vX_1^n\}_{n=1}^{\infty}}
		\!\!\!
		\left( \!
		\pliminf{n\to\infty}\!
		\frac{1}{n}{i}(\vX_1^n;\vY_1^n)
		-
		\plimsup{n\to\infty}\!
		\frac{1}{n}{i}(\vX_1^n;\vZ_1^n)
		\! \right)^{\!\!\!+}
	\end{align*}
	are achievable under the reliability criterion \eqref{equ:relcrit}
	and the secrecy criterion \eqref{equ:seccrit}. We leverage \cite[Lemma~2]{4797642} for deducing a lower bound
	on the achievable secure rates over ISI-WTCs.\footnote{Note that
		since ISI channels are indecomposable FSMCs~\cite{Gallager:1968:ITR:578869}
		(i.e., the effect of an initial state vanishes over time), the information
		rates are well-defined even if the initial state is unknown.}


Consider an ISI-WTC as in Definition~\ref{def:prwt:channel:1}. For all positive
integers $\ell$ and $\nu\geq\max(\mB,\mE)$, let $\bigl\{ \vX_{k(\ell+2\nu)-\nu+1}^{k(\ell+2\nu)+(\ell+\nu)}
\bigr\}_{k=-\infty}^{+\infty}$ be a block i.i.d.\ process where each block has
length $\ell + 2\nu$. So, it suffices to specify the distribution of a single
block $\vX_{k(\ell+2\nu)-\nu+1}^{k(\ell+2\nu)+(\ell+\nu)}$. In order to ensure
that there is no interference across blocks, we set
\begin{align*}
  X_{k(\ell+2\nu)+(\ell+1)}
    &\defeq
       0,
  \quad \ldots, \quad
  X_{k(\ell+2\nu)+(\ell+\nu)}
    \defeq
       0,
\end{align*}
while allowing
$\vX_{k(\ell+2\nu)-\nu+1}^{k(\ell+2\nu)+\ell}$ to be arbitrarily
distributed. Obviously, 
\begin{align*}
	\bigl\{
	\vX_{k(\ell+2\nu)-\nu+1}^{k(\ell+2\nu)+(\ell+\nu)},
	\vY_{k(\ell+2\nu)-\nu+1}^{k(\ell+2\nu)+(\ell+\nu)}
	\bigr\}_{k=-\infty}^{+\infty}
\end{align*}
is a joint block i.i.d.\ process. Similarly, \\
\begin{align*}
	\bigl\{
	\vX_{k(\ell+2\nu)-\nu+1}^{k(\ell+2\nu)+(\ell+\nu)},
	\vZ_{k(\ell+2\nu)-\nu+1}^{k(\ell+2\nu)+(\ell+\nu)}
	\bigr\}_{k=-\infty}^{+\infty}
\end{align*}
is also a joint block i.i.d.\ process. Let 
\begin{align*}
	\{X_t,Y_t\}_{t=1}^n
	= \bigl\{
	\vX_{k(\ell+2\nu)-\nu+1}^{k(\ell+2\nu)+(\ell+\nu)},
	\vY_{k(\ell+2\nu)-\nu+1}^{k(\ell+2\nu)+(\ell+\nu)}
	\bigr\}_{k=1}^{n\sp\prime\!}, 
\end{align*}
where $n\sp\prime$ denotes the number of i.i.d. blocks in $\{X_t\}_{t=1}^n$.
Then,
\begin{align*}
	&\lim_{n\to\infty}\frac{1}{n}i(\vX_1^n;\vY_1^n)\\
	&=n\sp\prime\cdot\lim_{n\to\infty}
	\frac{1}{n}
	\left(\!
	\frac{1}{n\sp\prime}\sum_{k=1}^{n\sp\prime\!}
	i\big(\vX_{k(\ell+2\nu)-\nu+1}^{k(\ell+2\nu)+(\ell+\nu)};
	\vY_{k(\ell+2\nu)-\nu+1}^{k(\ell+2\nu)+(\ell+\nu)}\big)\!\!\right)\\
	&=\frac{1}{\ell+2\nu}
	I(\vX_{-\nu+1}^{\ell+\nu}; \vY_{-\nu+1}^{\ell+\nu}), \qquad \text{w.p.\ $1$,}\alabel{equ:apn:prp1:lln1}
\end{align*}
where the second equality follows from the strong law of large
numbers and $n=n\sp\prime\cdot(\ell+2\nu)$. With an analogous manipulation, we
have
\begin{equation}\label{equ:apn:prp1:lln2}
	\lim_{n\to\infty}
	\frac{1}{n}
	i(\vX_1^n;\vZ_1^n)
	=\frac{1}{\ell+2\nu}
	I(\vX_{-\nu+1}^{\ell+\nu}; \vZ_{-\nu+1}^{\ell+\nu}),
	\qquad\text{w.p.~$1$}.
\end{equation}
Note
\begin{align*}
  I(\vX_{-\nu+1}^{\ell+\nu}; \vY_{-\nu+1}^{\ell+\nu})
    &\geq 
       I(\vX_{-\nu+1}^{\ell}; \vY_1^{\ell}) \\
    &= I(\vX_{-\nu+1}^{0}; \vY_1^\ell)
       +
       I(\vX_1^{\ell}; \vY_1^\ell | \vX_{-\nu+1}^0) \\
    &\geq I(\vX_1^{\ell}; \vY_1^{\ell} | \vX_{-\nu+1}^0).
\end{align*}
Moreover, 
\begin{align*}
  \hspace{0.5cm}&\hspace{-0.5cm}
  I(\vX_{-\nu+1}^{\ell+\nu}; \vZ_{-\nu+1}^{\ell+\nu}) \\
    &= I(\vX_{-\nu+1}^{\ell}; \vZ_{-\nu+1}^{\ell+\nu}) \\
    &= I(\vX_{-\nu+1}^{\ell}; \vZ_1^{\ell})
       +
       I(\vX_{-\nu+1}^{\ell}; \vZ_{-\nu+1}^{0} | \vZ_1^{\ell}) \\
    &\quad\ 
       + 
       I(\vX_{-\nu+1}^{\ell}; \vZ_{\ell+1}^{\ell+\nu}|\vZ_{-\nu+1}^{\ell}) \\
    &= I(\vX_1^{\ell}; \vZ_1^{\ell} | \vX_{-\nu+1}^{0})
       +
       I(\vX_{-\nu+1}^{0}; \vZ_1^{\ell}) \\
    &\quad\ 
       + I(\vX_{-\nu+1}^{\ell}; \vZ_{-\nu+1}^{0} | \vZ_1^{\ell})
       + I(\vX_{-\nu+1}^{\ell}; \vZ_{\ell+1}^{\ell+\nu} | \vZ_{-\nu+1}^{\ell}) \\
    &= I(\vX_1^{\ell}; \vZ_1^{\ell} | \vX_{-\nu+1}^{0})
       +
       I(\vX_{-\nu+1}^{0}; \vZ_1^{\ell}) \\
    &\quad\ 
       + 
       I(\vX_{-\nu+1}^{0}; \vZ_{-\nu+1}^{0} | \vZ_1^{\ell})
       +
       I(\vX_{\ell-\nu+1}^{\ell}; \vZ_{\ell+1}^{\ell+\nu} | \vZ_{-\nu+1}^{\ell}) \\
    &\leq 
       I(\vX_1^{\ell}; \vZ_1^{\ell} | \vX_{-\nu+1}^{0})
       +
       3\nu\log|\setX|.
\end{align*}
Combining \cite[Lemma~2]{4797642} with \eqref{equ:apn:prp1:lln1}, \eqref{equ:apn:prp1:lln2}, and the above lower and upper bounds
implies that all secure rates $\Rs$ satisfying
\begin{align*}
	\Rs
	&< \frac{1}{\ell+2\nu}
	\Big(
	I(\vX_1^{\ell};\vY_1^{\ell}|\vX_{-\nu+1}^{0})
	-
	I(\vX_1^{\ell};\vZ_1^{\ell}|\vX_{-\nu+1}^{0})  \\
	&\hspace{2cm}
	-
	3\nu \cdot \log|\setX|
	\Big)^+\alabel{inequ:achiev}
\end{align*}
are achievable under the reliability criterion~\eqref{equ:relcrit} and the secrecy criterion~\eqref{equ:seccrit}.

Let $\nu$ be the memory order of an FSMC associated with the considered ISI-WTC and let $\matrQ\in\setQ(\setB)$ be the
parameter of the input Markov source. Define $$\Rs(\matrQ)\defeq\lim_{n\to\infty}\frac{1}{n}\Big(I(\vS_1^n; \vY_1^n | S_0)-I(\vS_1^n; \vZ_1^n | S_0)\Big)^+.$$
It is easy to verify
\begin{align*}
I(\vX_1^n; \vY_1^n | \vX_{-\nu+1}^{0})
&= I(\vS_1^n; \vY_1^n | S_0),\\
I(\vX_1^n; \vZ_1^n | \vX_{-\nu+1}^{0})
&=I(\vS_1^n; \vZ_1^n | S_0).
\end{align*}
By letting $n\to\infty$ and invoking \eqref{inequ:achiev}, all secure rates $\Rs$ satisfying $\Rs<\Rs(\matrQ)$ are achievable.
Finally, reformulating the expression of $\Rs(\matr{Q})$ as follows proves the promised result.
	\begin{alignat*}{1}
		\Rs(\matrQ)
		&= \lim_{n\to\infty}
		\frac{1}{n}
		\sum_{t=1}^n
		\bigl(
		I(S_t; \vY_1^n | \vS_0^{t-1})
		-
		I(S_t; \vZ_1^n | \vS_0^{t-1})
		\bigr)\\
		&= \lim_{n\to\infty}
		\frac{1}{n}
		\sum_{t=1}^n
		\bigl(
		I(S_t; \vY_1^n | S_{t-1})
		-
		I(S_t; \vZ_1^n | S_{t-1})
		\bigr) \\
		&= \lim_{n\to\infty}
		\frac{1}{n}
		\sum_{t=1}^n
		\bigl(
		H(S_t | \vZ_1^n, S_{t-1})
		-
		H(S_t | \vY_1^n, S_{t-1})
		\bigr)\\
		&= \sum_{(i,j) \in \setB}
		Q_{ij}
		\cdot
		\bigl(
		\TB_{ij}(\matrQ)
		-
		\TE_{ij}(\matrQ)
		\bigr),
	\end{alignat*}
where the $(\cdot)^+$ operator has been omitted for clarity of the presentation and where the last equality is based on expressing $H(S_t|\vY_1^n,S_{t-1})$ as \eqref{equ:eqvv:Tij} at the top of this page,
	\bigformulatop{18}{19}{
		\begin{align*}
		H(S_t|\vY_1^n,S_{t-1})  &=-
		\sum_{(i,j)\in\setB}
		{\int_{\vy_1^n \in \setY^n}}
		p_{S_t, S_{t-1}, \vY_1^n}
		(j, i, \vy_1^n)
		\cdot 
		\log
		\bigl(
		p_{S_t | S_{t-1}, \vY_1^n}
		(j | i, \vy_1^n)
		\bigr){\mathrm{d}\vy_1^n} \nonumber \\
		&= -
		\sum_{(i,j)\in\setB}
		{\int_{\vy_1^n \in \setY^n}}
		p_{S_t, S_{t-1}, \vY_1^n}
		(j, i, \vy_1^n)
		\cdot 
		\left(
		\log
		\left(
		\frac{p_{S_t, S_{t-1}, \vY_1^n}
			(j, i, \vy_1^n)}
		{p_{\vY_1^n}(\vy_1^n)}
		\right)
		-
		\log
		\left(
		\frac{p_{S_{t-1}, \vY_1^n}
			(i, \vy_1^n)}
		{p_{\vY_1^n}(\vy_1^n)}
		\right)
		\right){\mathrm{d}\vy_1^n} \nonumber \\
		&= -
		\sum_{(i,j)\in\setB}
		\mu_i 
		p_{ij}
		\cdot
		{\int_{\vy_1^n \in \setY^n}}
		\Bigl(
		p_{\vY_1^n | S_{t-1}, S_t}
		(\vy_1^n | i, j)
		\cdot
		\log
		\bigl(
		p_{S_{t-1}, S_t|\vY_1^n}
		(i, j | \vy_1^n)
		\bigr) \\
		&\hspace{6cm}
		- 
		p_{\vY_1^n | S_{t-1}}
		(\vy_1^n | i)
		\cdot
		\log
		\bigl(
		p_{S_{t-1} | \vY_1^n}
		(i | \vy_1^n)
		\bigr)
		\Bigr){\mathrm{d}\vy_1^n} \nonumber \\
		&= -
		\sum_{(i,j)\in\setB}
		\mu_i
		p_{ij}
		\cdot
		{\int_{\vy_1^n \in \setY^n}}
		\Biggl(
		\frac{p_{S_{t-1}, S_t | \vY_1^n}
			(i, j | \vy_1^n)}
		{\mu_i p_{ij}}
		\cdot
		p_{\vY_1^n}(\vy_1^n)
		\cdot
		\log
		\bigl(
		p_{S_{t-1}, S_t | \vY_1^n}
		(i, j | \vy_1^n)
		\bigr) \nonumber \\
		&\hspace{6cm}
		-
		\frac{p_{S_{t-1} | \vY_1^n}
			(i | \vy_1^n)}
		{\mu_i}
		\cdot
		p_{\vY_1^n}(\vy_1^n)
		\cdot
		\log
		\bigl(
		p_{S_{t-1} | \vY_1^n}(i | \vy_1^n)
		\bigr)
		\Biggr){\mathrm{d}\vy_1^n} \nonumber \\
		&= -
		\sum_{(i,j)\in\setB}
		\mu_i
		p_{ij}
		\cdot
		\left(
		{\int_{\vy_1^n \in \setY^n}}
		p_{\vY_1^n}(\vy_1^n)
		\cdot
		\log
		\left(
		\frac{p_{S_{t-1}, S_t | \vY_1^n}
			(i, j | \vy_1^n)
			^{{p_{S_{t-1}, S_t | \vY_1^n}
					(i, j | \vy_1^n)} 
				/ \mu_i p_{ij}}}
		{p_{S_{t-1} | \vY_1^n}
			(i | \vy_1^n)
			^{{p_{S_{t-1} | \vY_1^n}
					(i | \vy_1^n)}
				/ \mu_i}}
		\right)
		\right){\mathrm{d}\vy_1^n}.\alabel{equ:eqvv:Tij}
	\end{align*}}
with an
analogous expression for $H(S_t|\vZ_1^n,S_{t-1})$, along with
using~\eqref{eq:def:TB:1} and~\eqref{eq:def:TE:1}.


\section{Proof of Lemma~\ref{lem.convexity}}\label{apndx.convexity}


Besides the assumptions on the parameterizations $\matrQ(\theta)$ made in
Remark~\ref{remark:Q:parameterized:family:1}, we will also assume that for all
$(i,j) \in \setB$, the functions $Q_{ij}(\theta)$ and $\mu_i(\theta)$ are
affine functions in terms of~$\theta$, which implies
$Q_{ij}^{\theta\theta}(\theta) = 0$ and $\mu_i^{\theta\theta}(\theta) =0$,
where the superscript $\theta\theta$ denotes the second-order derivative
w.r.t.\ $\theta$.


Denoting the second-order derivative of $\bpsi_{\tmatrQ}(\theta)$ by
$\bpsi_{\tmatrQ}^{\theta\theta}(\theta)$, we observe that the claim in the
lemma statement is equivalent to
$\bpsi_{\tmatrQ}^{\theta\theta}(\theta) \geq 0$ for all possible
parameterizations of $\matrQ(\theta)$ that satisfy the above-mentioned
conditions.


Let $\hQ_{ij} \defeq (1-\kappa) \cdot \tQ_{ij} + \kappa \cdot Q_{ij}$ and $\hmu_i \defeq (1-\kappa) \cdot \tmu_i + \kappa \cdot Q_{ij}$ for all $(i,j)\in\setB$. Some straightforward calculations show that
\begin{align*}
	\bpsi_{\tmatrQ}^{\theta\theta}(\theta) &= \kappa^2 \kappa' \cdot \left( \sum_{(i,j) \in \setB} \frac{(Q^{\theta}_{ij})^2}{\hat{Q}_{ij}} - \sum_{i \in \setS} \frac{(\mu_i^\theta)^2}{\hat{\mu}_i} \right)\\ &= \kappa^2 \kappa' \cdot \sum_{i \in \setS} \left( \left( \sum_{j \in \setSright{i}} \frac{(Q^{\theta}_{ij})^2}{\hat{Q}_{ij}} \right) - \frac{(\mu_i^\theta)^2}{\hat{\mu}_i} \right).
\end{align*}
Noting that for any $i \in \setS$ it holds that
\begin{align*} \sum_{j \in \setSright{i}} \frac{(Q^{\theta}_{ij})^2}{\hat{Q}_{ij}}  &= \hat{\mu}_i \cdot \sum_{j \in \setSright{i}} \frac{\hat{Q}_{ij}}{\hat{\mu}_i} \cdot \left( \frac{Q^{\theta}_{ij}}{\hat{Q}_{ij}} \right)^2\\ &\geq \hat{\mu}_i \cdot \left( \sum_{j \in \setSright{i}} \frac{\hat{Q}_{ij}}{\hat{\mu}_i} \cdot \frac{Q^{\theta}_{ij}}{\hat{Q}_{ij}} \right)^2\\ &= \frac{1}{\hat{\mu}_i} \cdot \left( \sum_{j \in \setSright{i}} Q^{\theta}_{ij} \right)^2 = \frac{(\mu_i^\theta)^2}{\hat{\mu}_i},
\end{align*}
where the inequality follows from Jensen's inequality, we can conclude that, indeed,
$\bpsi_{\tmatrQ}^{\theta\theta}(\theta) \geq 0$.


\section{Proof of Proposition~\ref{prop.DistEig}}\label{apndx.bee}


Maximizing $\psi_{\tmatrQ}(\matrQ)$ over $\matrQ \in \setQ(\setB)$ means to
optimize a differentiable, concave function over a polytope. We therefore set
up the Lagrangian
\begin{align*}
	L
	&\defeq \sum_{(i,j)\in\setB}
	Q_{ij} 
	\cdot
	\bigl(
	\tTB_{ij}
	-
	\tTE_{ij}
	\bigr)
	-
	\bpsi_{\tmatrQ}(\matrQ)\\
	&+ 
	\lambda 
	\cdot 
	\left(
	\sum_{(i,j)\in\setB}
	\!\!\! 
	Q_{ij}
	-
	1
	\right)
	+
	\sum_{(i,j)\in\setB}
	\!\!\! 
	\lambda_j
	Q_{ij}
	-
	\sum_{(i,j)\in\setB}
	\!\!\! 
	\lambda_i
	Q_{ij}.
\end{align*}
Note that at this stage we omit Lagrangian multipliers w.r.t.\ the
constraints $Q_{ij} \geq 0$, $(i,j) \in \setB$. We will make sure at a later
stage that these constraints are satisfied thanks to the choice of $\kappa$
in~\eqref{equ.kappalwrbound}.


Recall that we assume that the surrogate function takes on its maximal value
at $\matrQ = \matrQ^{*}$. Therefore, setting the gradient of $L$ equal to the
zero vector at $\matrQ = \matrQ^{*}$, we obtain
\begin{alignat}{2}
	\hspace{-0.15cm}
	0
	&= \left.
	\frac{\partial L}
	{\partial Q_{ij}}
	\right|_{\matrQ = \matrQ^*} \hspace{-0.50cm}
	&&= \tTB_{ij} 
	- 
	\tTE_{ij}
	-
	\left.
	\frac{\partial \bpsi_{\tmatrQ}(\matrQ)}
	{\partial Q_{ij}}
	\right|_{\matrQ = \matrQ^*} \hspace{-0.50cm}
	+
	\lambda^{*}
	+
	\lambda_j^{*}
	-
	\lambda_i^{*},
	\nonumber \\
	& &&\hspace{5.5cm} (i,j) \in \setB,
	\label{eq:Lagrangian:gradient:1} \\
	\hspace{-0.15cm}
	0
	&= \left.
	\frac{\partial L}
	{\partial \lambda}
	\right|_{\matrQ = \matrQ^*} \hspace{-0.25cm}
	&&= \sum_{(i,j) \in \setB}
	Q^{*}_{ij}
	-
	1, \nonumber \\
	\hspace{-0.15cm}
	0
	&= \left.
	\frac{\partial L}
	{\partial \lambda_i}
	\right|_{\matrQ = \matrQ^*} \hspace{-0.25cm}
	&&= \sum_{r \in \setSleft{i}}
	Q_{ri}^{*}
	-
	\sum_{j \in \setSright{i}}
	Q^{*}_{ij},
	\hspace{2.5cm} i \in \setS, \nonumber
\end{alignat}
where
\begin{align}
	&\left.
	\frac{\partial\bpsi_{\tmatrQ}(\matrQ)}{\partial Q_{ij}}
	\right|_{\matrQ = \matrQ^*}\nonumber\\
	&\,= \left.
	\kappa'
	\cdot
	\Bigl(
	\kappa
	{\cdot}
	\log
	\bigl(
	1
	+
	\kappa 
	{\cdot}
	(\deltaQ)_{ij} 
	\bigr)
	-
	\kappa
	{\cdot}
	\log
	\bigl(
	1
	+
	\kappa
	{\cdot}
	(\deltamu)_i
	\bigr)
	\Bigr)
	\right|_{\matrQ = \matrQ^*}
	\nonumber \\
	&\,= \kappa
	\cdot
	\kappa'
	\cdot
	\log
	\left(
	\frac{(1-\kappa) \cdot \tQ_{ij} + \kappa \cdot Q^{*}_{ij}}
	{(1-\kappa) \cdot \tmu_i 
		+
		\kappa \cdot \mu^{*}_i}
	\cdot
	\frac{\tmu_i}
	{\tQ_{ij}}
	\right) \nonumber \\
	&\,= \kappa
	\cdot
	\kappa'
	\cdot
	\log
	\left(
	\frac{\hQ^{*}_{ij}}
	{\hmu_i^{*}}
	\cdot
	\frac{\tmu_i}
	{\tQ_{ij}}
	\right) \nonumber \\
	&\,= \kappa
	\cdot
	\kappa'
	\cdot
	\log( \hp^{*}_{ij} )
	-
	\kappa
	\cdot
	\kappa'
	\cdot
	\log( \tp_{ij} ).
	\label{eq:psibar:derivative:1}
\end{align}
Here the third and the fourth equality use
$\{ \hQ^{*}_{ij} \}_{(i,j) \in \setB}$, which is defined by
\begin{align}
	\hQ^{*}_{ij}
	&\defeq
	(1-\kappa)
	\cdot
	\tQ_{ij}
	+
	\kappa 
	\cdot
	Q^{*}_{ij},
	\quad (i,j) \in \setB,
	\label{eq:hat:Q:star:1}
\end{align}
along with $\{ \hmu^{*}_i \}_{i \in \setS}$ and
$\{ \hp_{ij}^{*} \}_{(i,j) \in \setB}$, which are derived from
$\{ \hQ^{*}_{ij} \}_{(i,j) \in \setB}$ in the usual manner. {Note that
	$\hmu^{*}_i \defeq \sum_{j' \in \setSright{i}} \hQ^{*}_{ij'} = (1-\kappa)
	\cdot \tmu_i + \kappa\cdot \mu_i^{*}$, for all $i \in \setS$, and}
\begin{align}
	\hp^{*}_{ij}
	&= \frac{\hQ^{*}_{ij}}
	{\hmu^{*}_i}
	= \frac{(1-\kappa) \cdot \tQ_{ij} + \kappa \cdot Q^{*}_{ij}}
	{(1-\kappa) \cdot \tmu_i + \kappa \cdot \mu_i^{*}} \nonumber \\
	&= \frac{(1-\kappa) \cdot \tQ_{ij} + \kappa \cdot Q^{*}_{ij}}
	{(1-\kappa) \cdot \tmu_i 
		+ 
		\kappa \cdot \sum_{j' \in \setSright{i}}Q_{ij'}^{*}},
	\quad (i,j) \in \setB.
	\label{eq:hat:p:star:1}
\end{align}
Note also that solving~\eqref{eq:hat:Q:star:1} for $Q^{*}_{ij}$ results in
\begin{align*}
Q^{*}_{ij}
= \frac{1}{\kappa}
\cdot
(
\hQ^{*}_{ij}
-
\tQ_{ij}
+
\kappa
\cdot
\tQ_{ij}
), \quad 
(i,j) \in \setB,
\end{align*}
which shows that $Q^{*}_{ij} \geq 0$, $(i,j) \in \setB$, for $\kappa$
satisfying~\eqref{equ.kappalwrbound}. (Recall that when setting up the
Lagrangian, we omitted the Lagrange multipliers for the constraints
$Q_{ij} \geq 0$, $(i,j) \in \setB$; therefore we have to verify that the
solution satisfies these constraints, which it does indeed.)


Combining~\eqref{eq:Lagrangian:gradient:1} and~\eqref{eq:psibar:derivative:1},
and solving for $\hp^{*}_{ij}$ results in
\begin{align*}
	\hp^{*}_{ij}
	= \tp_{ij}
	\cdot
	\exp
	\left(
	\frac{
		\tTB_{ij}
		-
		\tTE_{ij}
		+
		\lambda^{*}
		+
		\lambda_j^{*}
		-
		\lambda_i^{*}
	}
	{\kappa \kappa'}
	\right), \quad
	(i,j) \in \setB.
\end{align*}
Using~\eqref{equ.Aij} and defining
$\rho \defeq \exp\bigl( - \frac{\lambda^{*}}{\kappa \kappa'} \bigr)$ and
$\vgamma = \bigl(\gamma_i\defeq \exp\bigl( \frac{\lambda_i^{*}}{\kappa
	\kappa'} \bigr)\bigr)_{i \in \setS}$, we rewrite this equation
as
 \begin{align*}
   \hp^{*}_{ij}
      = \frac{A_{ij}}
             {\rho}
        \cdot
        \frac{\gamma_j}
             {\gamma_i},
             \quad
         (i,j) \in \setB.
 \end{align*}
Because $\sum_{j \in \setSright{i}} \hp^{*}_{ij} = 1$ for all $i \in \setS$,
summing both sides of this equation over $j \in \setSright{i}$ results in
\begin{align*}
	1
	= \sum_{j \in \setSright{i}}
	\frac{A_{ij}}
	{\rho}
	\cdot
	\frac{\gamma_j}
	{\gamma_i},
	 \quad
	i \in \setS,
\end{align*}
or, equivalently,
\begin{align*}
	\rho
	\cdot 
	\gamma_i
	= \sum_{j \in \setSright{i}}
	A_{ij}
	\cdot
	\gamma_j,
	 \quad 
	i \in \setS.
\end{align*}
This system of linear equations can be written as
\begin{align*}
	\matrA \cdot \vgamma 
	= \rho \cdot \vgamma.
\end{align*}
Clearly, this equation can only be satisfied if $\vgamma$ is an eigenvector of
$\matrA$ with corresponding eigenvalue $\rho$. A slightly lengthy calculation
(which is somewhat similar to the calculation in~\cite[Eq.~(51)]{4494705})
shows that
\begin{align}
\psi_{\tmatrQ}(\matrQ^{*})
= \log(\rho).
\label{eq:maximum:of:surrogate:function:1}
\end{align}
Clearly, in order to maximize the right-hand side of~\eqref{eq:maximum:of:surrogate:function:1}
over all eigenvalues of
$\matrA$, the eigenvalue $\rho$ has to be the Perron--Frobenius eigenvalue and
$\vgamma$ the corresponding eigenvector.


The proof is concluded by noting that~\eqref{eq:hat:p:star:1} can be rewritten
as the system of linear equations
 \begin{align*}
   Q^{*}_{ij}
   -
   \hp_{ij}^{*}
     \cdot
     \sum_{j' \in \setSright{i}}
       Q_{ij'}^{*}
      = \frac{1-\kappa}
             {\kappa}
        \cdot
        \bigl(
          \tmu_i \hp^{*}_{ij}
          -
          \tQ_{ij}
        \bigr), \quad
     (i,j) \in \setB,
 \end{align*}
which can be used to determine $\{ Q^{*}_{ij} \}_{(i,j) \in \setB}$, because
all other quantities appearing in these equations are either known or have
already been calculated.


\bibliographystyle{IEEEtran}
\bibliography{citation}\vspace*{12pt}

\begin{thebibliography}{10}
\providecommand{\url}[1]{#1}
\csname url@samestyle\endcsname
\providecommand{\newblock}{\relax}
\providecommand{\bibinfo}[2]{#2}
\providecommand{\BIBentrySTDinterwordspacing}{\spaceskip=0pt\relax}
\providecommand{\BIBentryALTinterwordstretchfactor}{4}
\providecommand{\BIBentryALTinterwordspacing}{\spaceskip=\fontdimen2\font plus
\BIBentryALTinterwordstretchfactor\fontdimen3\font minus
  \fontdimen4\font\relax}
\providecommand{\BIBforeignlanguage}[2]{{%
\expandafter\ifx\csname l@#1\endcsname\relax
\typeout{** WARNING: IEEEtran.bst: No hyphenation pattern has been}%
\typeout{** loaded for the language `#1'. Using the pattern for}%
\typeout{** the default language instead.}%
\else
\language=\csname l@#1\endcsname
\fi
#2}}
\providecommand{\BIBdecl}{\relax}
\BIBdecl

\bibitem{Nouri:Asvadi:Chen:Vontobel:21:1}
A.~Nouri, R.~Asvadi, J.~Chen, and P.~O. Vontobel, ``Finite-input intersymbol
  interference wiretap channels,'' in \emph{Proc.\ IEEE Inf.\ Theory Workshop},
  Kanazawa, Japan, Oct.~17--21 2021, pp. 1--6.

\bibitem{7539590}
Y.~Liu, H.~H. Chen, and L.~Wang, ``Physical layer security for next generation
  wireless networks: Theories, technologies, and challenges,'' \emph{IEEE
  Commun.\ Surv.\ Tutor.}, vol.~19, no.~1, pp. 347--376, 1st Quart., 2017.

\bibitem{Gidney_2021}
C.~Gidney and M.~Eker{\aa}, ``How to factor 2048 bit \text{RSA} integers in 8
  hours using 20 million noisy qubits,'' \emph{Quantum}, vol.~5, p. 433, Apr.
  2021.

\bibitem{9380147}
M.~Bloch, O.~{G\"unl\"u}, A.~Yener, F.~Oggier, H.~V. Poor, L.~Sankar, and R.~F.
  Schaefer, ``An overview of information-theoretic security and privacy:
  metrics, limits and applications,'' \emph{IEEE J.\ Sel.\ Areas Inf.\ Theory},
  vol.~2, no.~1, pp. 5--22, Mar. 2021.

\bibitem{8509094}
J.~M. Hamamreh, H.~M. Furqan, and H.~Arslan, ``Classifications and applications
  of physical layer security techniques for confidentiality: A comprehensive
  survey,'' \emph{IEEE Commun.\ Surv.\ Tutor.}, vol.~21, no.~2, pp. 1773--1828,
  2nd Quart., 2019.

\bibitem{Proakis2008}
J.~G. Proakis and M.~Salehi, \emph{Digital Communications}, 5th~ed.\hskip 1em
  plus 0.5em minus 0.4em\relax McGraw Hill, 2008.

\bibitem{7931557}
L.~Zhang, A.~Ijaz, P.~Xiao, and R.~Tafazolli, ``Channel equalization and
  interference analysis for uplink narrowband internet of things
  (\text{NB-IoT}),'' \emph{IEEE Commun.\ Lett.}, vol.~21, no.~10, pp.
  2206--2209, May 2017.

\bibitem{8698792}
J.~Choi, ``Single-carrier index modulation for \text{IoT} uplink,'' \emph{IEEE
  J.\ Sel.\ Top.\ Signal Process.}, vol.~13, no.~6, pp. 1237--1248, Oct. 2019.

\bibitem{3GPPTS36211}
\text{European Telecommunications Standards Institute}, ``Evolved universal
  terrestrial radio access (\text{E-UTRAN}): Physical channels and
  modulation,'' \emph{ETSI TS \emph{136 211 V16.5.0}}, May 2021.

\bibitem{ericssonwpp:IoT}
C.~Kuhlins, B.~Rathonyi, A.~Zaidi, and M.~Hogan, ``Cellular networks for
  massive \text{IoT},'' \emph{Ericsson White Paper \emph{Uen 284 23-3278}},
  Jan. 2020.

\bibitem{9246302}
Y.~Cao, W.~Shi, L.~Sun, and X.~Fu, ``Channel state information-based ranging
  for underwater acoustic sensor networks,'' \emph{IEEE Trans.\ Wirel.\
  Commun.}, vol.~20, no.~2, pp. 1293--1307, Feb. 2021.

\bibitem{1569979}
M.~Sahin and H.~Arslan, ``Inter-symbol interference in high data rate
  \textmd{UWB} communications using energy detector receivers,'' in \emph{Proc.
  IEEE Int. Conf. on Ultra-Wideband}, Zurich, Switzerland, Sept. 2005, pp.
  176--179.

\bibitem{8856252}
B.~Dai, Z.~Ma, Y.~Luo, X.~Liu, Z.~Zhuang, and M.~Xiao, ``Enhancing physical
  layer security in internet of things via feedback: A general framework,''
  \emph{IEEE Internet Things J.}, vol.~7, no.~1, pp. 99--115, Jan. 2020.

\bibitem{8715341}
J.~Zhang, S.~Rajendran, Z.~Sun, R.~Woods, and L.~Hanzo, ``Physical layer
  security for the internet of things: Authentication and key generation,''
  \emph{IEEE Wirel.\ Commun.}, vol.~26, no.~5, pp. 92--98, Oct. 2019.

\bibitem{8428404}
S.~Jiang, ``On securing underwater acoustic networks: A survey,'' \emph{IEEE
  Commun.\ Surv.\ Tutor.}, vol.~21, no.~1, pp. 729--752, 1st Quart., 2019.

\bibitem{3GPPTS36101}
\text{European Telecommunications Standards Institute}, ``\text{LTE}; evolved
  universal terrestrial radio access (\text{E-UTRA}): User equipment
  (\text{UE}) radio transmission and reception,'' \emph{ETSI TS 136 101 V16.9.0
  Release 16}, May 2021.

\bibitem{Ma-2018}
J.~Ma, R.~Shrestha, J.~Adelberg, C.~Y. Yeh, Z.~Hossain, E.~Knightly, J.~M.
  Jornet, and D.~M. Mittleman, ``Security and eavesdropping in terahertz
  wireless links,'' \emph{Nature}, vol. 563, pp. 89--93, Oct. 2018.

\bibitem{8861076}
T.~M. Duman and M.~Stojanovic, ``Information rates of energy harvesting
  communications with intersymbol interference,'' \emph{IEEE Commun.\ Lett.},
  vol.~23, no.~12, pp. 2164--2167, Dec. 2019.

\bibitem{Gallager:1968:ITR:578869}
R.~G. Gallager, \emph{Information Theory and Reliable Communication}.\hskip 1em
  plus 0.5em minus 0.4em\relax New~York, NY, USA: John Wiley \& Sons, 1968.

\bibitem{1054855}
R.~Blahut, ``Computation of channel capacity and rate-distortion functions,''
  \emph{IEEE Trans.\ Inf.\ Theory}, vol.~18, no.~4, pp. 460--473, July 1972.

\bibitem{1054753}
S.~Arimoto, ``An algorithm for computing the capacity of arbitrary discrete
  memoryless channels,'' \emph{IEEE Trans.\ Inf.\ Theory}, vol.~18, no.~1, pp.
  14--20, Jan. 1972.

\bibitem{4494705}
P.~O. Vontobel, A.~Kavčić, D.~M. Arnold, and H.~A. Loeliger, ``A
  generalization of the \text{Blahut-Arimoto} algorithm to finite-state
  channels,'' \emph{IEEE Trans.\ Inf.\ Theory}, vol.~54, no.~5, pp. 1887--1918,
  May 2008.

\bibitem{965977}
A.~Kavčić, ``On the capacity of {Markov} sources over noisy channels,'' in
  \emph{Proc.\ IEEE Glob.\ Commun.\ Conf.}, vol.~5, San Antonio, TX, USA, Nov.
  2001, pp. 2997--3001.

\bibitem{1397923}
S.~Yang, A.~Kavčić, and S.~Tatikonda, ``Feedback capacity of finite-state
  machine channels,'' \emph{IEEE Trans.\ Inf.\ Theory}, vol.~51, no.~3, pp.
  799--810, Mar. 2005.

\bibitem{955166}
P.~O. Vontobel and D.~M. Arnold, ``An upper bound on the capacity of channels
  with memory and constraint input,'' in \emph{Proc.\ IEEE Inf.\ Theory
  Workshop}, Cairns, Queensland, Australia, Sept. 2001, pp. 147--149.

\bibitem{4455735}
J.~Chen and P.~H. Siegel, ``Markov processes asymptotically achieve the
  capacity of finite-state intersymbol interference channels,'' \emph{IEEE
  Trans.\ Inf.\ Theory}, vol.~54, no.~3, pp. 1295--1303, Mar. 2008.

\bibitem{8747418}
T.~S. Han and M.~Sasaki, ``Wiretap channels with causal state information:
  Strong secrecy,'' \emph{IEEE Trans.\ Inf.\ Theory}, vol.~65, no.~10, pp.
  6750--6765, Oct. 2019.

\bibitem{9483917}
------, ``Wiretap channels with causal and non-causal state information:
  Revisited,'' \emph{IEEE Trans.\ Inf.\ Theory}, vol.~67, no.~9, pp. 6122 --
  6139, Sept. 2021.

\bibitem{8963770}
B.~Dai, C.~Li, Y.~Liang, Z.~Ma, and S.~{Shamai (Shitz)}, ``Impact of
  action-dependent state and channel feedback on \text{Gaussian} wiretap
  channels,'' \emph{IEEE Trans.\ Inf.\ Theory}, vol.~66, no.~6, pp. 3435--3455,
  June 2020.

\bibitem{7774989}
B.~Dai, Z.~Ma, and Y.~Luo, ``Finite state {Markov} wiretap channel with delayed
  feedback,'' \emph{IEEE Trans.\ Inf.\ Forensics Secur.}, vol.~12, no.~3, pp.
  746--760, Mar. 2017.

\bibitem{8464866}
S.~{Hanoglu}, S.~R. {Aghdam}, and T.~M. {Duman}, ``Artificial-noise-aided
  secure transmission over finite-input intersymbol interference channels,'' in
  \emph{Proc.\ 25th Int.\ Conf.\ Telecommun.}, Saint-Malo, France, June 2018,
  pp. 346--350.

\bibitem{8737786}
J.~{de Dieu Mutangana} and R.~Tandon, ``Blind \text{MIMO} cooperative jamming:
  secrecy via \text{ISI} heterogeneity without \text{CSIT},'' \emph{IEEE
  Trans.\ Inf.\ Forensics Secur.}, vol.~15, pp. 447--461, June 2020.

\bibitem{5351376}
Y.~Sankarasubramaniam, A.~Thangaraj, and K.~Viswanathan, ``Finite-state wiretap
  channels: Secrecy under memory constraints,'' in \emph{Proc.\ IEEE Inf.\
  Theory Workshop}, Taormina, Italy, Oct. 2009, pp. 115--119.

\bibitem{1055892}
I.~Csiszár and J.~K{\"o}rner, ``Broadcast channels with confidential
  messages,'' \emph{IEEE Trans.\ Inf.\ Theory}, vol.~24, no.~3, pp. 339--348,
  May 1978.

\bibitem{Nouri:Asvadi:22:ISIT}
A.~Nouri and R.~Asvadi, ``Matched information rate codes for binary-input
  intersymbol interference wiretap channels,'' in \emph{Proc.\ IEEE Int.\
  Symp.\ Inf.\ Theory}, Espoo, Finland, June 2022, pp. 1163--1168.

\bibitem{1661831}
D.~M. Arnold, H.~A. Loeliger, P.~O. Vontobel, A.~Kavčić, and W.~Zeng,
  ``Simulation-based computation of information rates for channels with
  memory,'' \emph{IEEE Trans.\ Inf.\ Theory}, vol.~52, no.~8, pp. 3498--3508,
  Aug. 2006.

\bibitem{1397934}
A.~{Kavčić}, X.~{Ma}, and N.~{Varnica}, ``Matched information rate codes for
  partial response channels,'' \emph{IEEE Trans.\ Inf.\ Theory}, vol.~51,
  no.~3, pp. 973--989, Mar. 2005.

\bibitem{4777638}
P.~Sadeghi, P.~O. Vontobel, and R.~Shams, ``Optimization of information rate
  upper and lower bounds for channels with memory,'' \emph{IEEE Trans.\ Inf.\
  Theory}, vol.~55, no.~2, pp. 663--688, Feb. 2009.

\bibitem{dempster1977maximum}
A.~P. Dempster, N.~M. Laird, and D.~B. Rubin, ``Maximum likelihood from
  incomplete data via the \text{EM} algorithm,'' \emph{J.\ R.\ Stat.\ Soc.\
  Ser.\ B (Methodological)}, pp. 1--38, 1977.

\bibitem{wu83}
C.~F.~J. Wu, ``On the convergence properties of the \text{EM} algorithm,''
  \emph{Ann.\ Stat.}, vol.~11, no.~1, pp. 95--103, Mar. 1983.

\bibitem{93141079}
G.~L. Stüber, \emph{Principles of Mobile Communication}, 4th~ed.\hskip 1em
  plus 0.5em minus 0.4em\relax Cham, Switzerland: Springer International
  Publishing, 2017.

\bibitem{We-2003}
W.~Xiang and S.~Pietrobon, ``On the capacity and normalization of {ISI}
  channels,'' \emph{IEEE Trans.\ Inf.\ Theory}, vol.~49, no.~9, pp. 2263--2268,
  Sept. 2003.

\bibitem{8922625}
A.~Chakrapani, ``\text{NB-IoT} uplink receiver design and performance study,''
  \emph{IEEE Internet Things J.}, vol.~7, no.~3, pp. 2469--2482, Mar. 2020.

\bibitem{RePEc:mtp:titles:0262100711}
K.~L. Judd, \emph{Numerical Methods in Economics}.\hskip 1em plus 0.5em minus
  0.4em\relax London, UK: The MIT Press, 1998.

\bibitem{ISIWTC_SF}
\BIBentryALTinterwordspacing
A.~Nouri, ``{ISI Wiretap Channels [SIMULATION\_FILES]},'' Oct. 2021. [Online].
  Available: \url{https://doi.org/10.5281/zenodo.5595240}
\BIBentrySTDinterwordspacing

\bibitem{LeungYanCheong:Hellman:78:1}
S.~K. Leung-Yan-Cheong and M.~E. Hellman, ``The {Gaussian} wire-tap channel,''
  \emph{IEEE Trans.\ Inf.\ Theory}, vol.~24, no.~4, pp. 451--456, Jul. 1978.

\bibitem{Maurer1994}
U.~Maurer, ``Communications and cryptography: Two sides of one tapestry,''
  R.~E. Blahut, D.~J. Costello, U.~Maurer, and T.~Mittelholzer, Eds.\hskip 1em
  plus 0.5em minus 0.4em\relax Springer, Boston, MA, USA: The Springer
  International Series in Engineering and Computer Science, 1994, vol. 276, pp.
  271--285.

\bibitem{6772207}
A.~D. Wyner, ``The wire-tap channel,'' \emph{Bell Syst.\ Tech.\ J.}, vol.~54,
  no.~8, pp. 1355--1387, Oct. 1975.

\bibitem{Bloch:2011:PSI:2829193}
M.~Bloch and J.~Barros, \emph{Physical-Layer Security: From Information Theory
  to Security Engineering}, 1st~ed.\hskip 1em plus 0.5em minus 0.4em\relax
  New~York, NY, USA: Cambridge University Press, 2011.

\bibitem{6612711}
M.~{Bloch} and J.~N. {Laneman}, ``Strong secrecy from channel resolvability,''
  \emph{IEEE Trans.\ Inf.\ Theory}, vol.~59, no.~12, pp. 8077--8098, Dec. 2013.

\bibitem{10.1007/978-3-642-32009-5_18}
M.~Bellare, S.~Tessaro, and A.~Vardy, ``Semantic security for the wiretap
  channel,'' in \emph{Proc.\ CRYPTO 2012}, vol. 7417, Berlin, Heidelberg, 2012,
  pp. 294--311.

\bibitem{1386546}
T.~S. Han, ``Folklore in source coding: information-spectrum approach,''
  \emph{IEEE Trans.\ Inf.\ Theory}, vol.~51, no.~2, pp. 747--753, Feb. 2005.

\bibitem{spectral}
------, \emph{Information-Spectrum Methods in Information Theory}.\hskip 1em
  plus 0.5em minus 0.4em\relax Berlin, Heidelberg, New~York: Springer, 2003.

\bibitem{4797642}
M.~Bloch and J.~N. Laneman, ``On the secrecy capacity of arbitrary wiretap
  channels,'' in \emph{Proc.\ 46th Annual Allerton Conf. Commun. Control and
  Computing}, Monticello, IL, USA, Sept. 2008, pp. 818--825.

\end{thebibliography}

\ifCLASSOPTIONcaptionsoff
\newpage
\fi
\vskip -2\baselineskip plus -1fil
\begin{IEEEbiography}[{\includegraphics[width=1in,height=1.25in,clip,keepaspectratio]{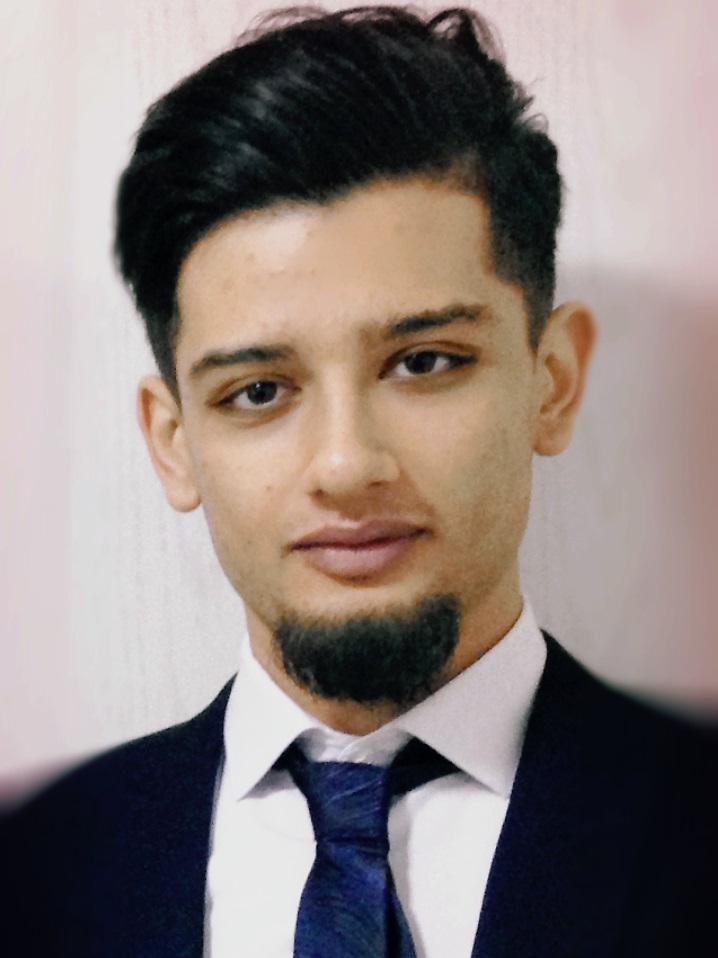}}]{Aria Nouri} (Graduate Student Member, IEEE) \text{received} his M.Sc.\ degree in electrical engineering, communications, from Shahid-Beheshti University, Tehran, Iran, in 2021. He has been with the cognitive telecommunication research group at Shahid-Beheshti University as a research assistant since 2017. His research interests lie in the areas of coding and information theory, focusing on secure communication, semantic communication, quantum error correction,\, and\, fault-tolerant\, quantum\, computing.
\end{IEEEbiography}
\begin{IEEEbiography}[{\includegraphics[width=1in,height=1.25in,clip,keepaspectratio]{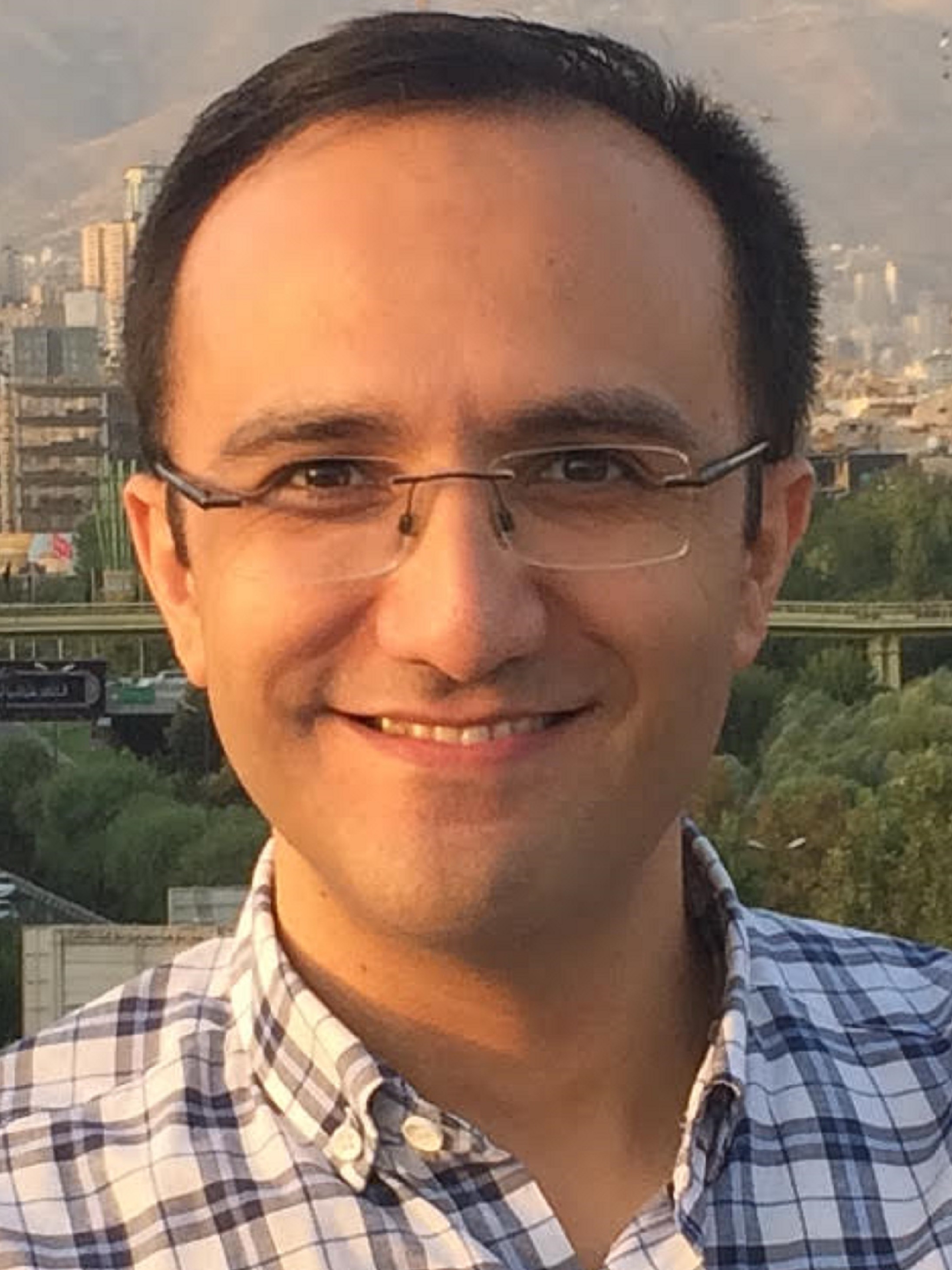}}]{Reza Asvadi}
	(Senior Member, IEEE) received the B.Sc. degree (with highest honors) in electrical engineering from K.\ N.\ Toosi University of Technology, Tehran, Iran, in 2001, the M.Sc. degree in electrical engineering from Sharif University of Technology, Tehran, Iran, in 2003, and the Ph.D. degree from K.\ N.\ Toosi University of Technology, in 2011.
	
	From 2004 to 2006, he was a Lecturer with Army Airforce University, Tehran, to fulfill his national service. He was a Post-Doctoral Researcher with the University of Oulu, Oulu, Finland, from 2012 to 2014. During the postdoc, he participated in many Academy of Finland and European Union (FP7) projects investigating iterative algorithms and information-theoretical bounds over new emerging wireless networks. He is currently an Assistant Professor with Shahid Beheshti University, Tehran, since 2016. His research interests include coding and information theory and signal processing for wireless communications.
	
	Dr. Asvadi was a recipient of several post-doctoral research grants, including the University of Alberta (2011–2012) and Carleton University (2014–2016) Postdoctoral Fellowships.
\end{IEEEbiography}

\begin{IEEEbiography}[{\includegraphics[width=1in,height=1.25in,clip,keepaspectratio]{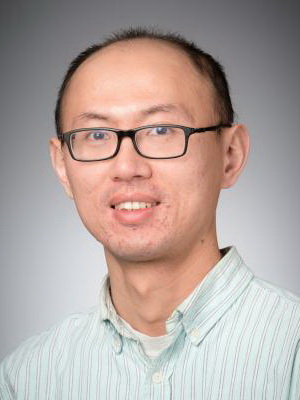}}]{Jun Chen}
	(Senior Member, IEEE) received the B.E. degree in communication engineering from Shanghai Jiao Tong University, Shanghai, China, in 2001, and the M.S. and Ph.D. degrees in electrical and computer engineering from Cornell University, Ithaca, NY, USA, in 2004 and 2006, respectively.
	
	From September 2005 to July 2006, he was a Post-Doctoral Research Associate with the Coordinated Science Laboratory, University of Illinois at Urbana–Champaign, Urbana, IL, USA, and a Post-Doctoral Fellow with the IBM Thomas J. Watson Research Center, Yorktown Heights, NY, USA, from July 2006 to August 2007. Since September 2007, he has been with the Department of Electrical and Computer Engineering, McMaster University, Hamilton, ON, Canada, where he is currently a Professor. His research interests include information theory, machine learning, wireless communications, and signal processing.
	
	Dr. Chen was a recipient of the Josef Raviv Memorial Postdoctoral Fellowship in 2006, the Early Researcher Award from the Province of Ontario in 2010, the IBM Faculty Award in 2010, the ICC Best Paper Award in 2020, and the JSPS Invitational Fellowship in 2021. He held the title of the Barber-Gennum Chair of information technology from 2008 to 2013 and the title of the Joseph Ip Distinguished Engineering Fellow from 2016 to 2018. He served as an Editor for {\sc IEEE Transactions on Green Communications and Networking} from 2020 to 2021. He is currently an Associate Editor of {\sc IEEE Transactions on Information Theory}.
\end{IEEEbiography}

\newpage

\begin{IEEEbiography}[{\includegraphics[width=1in,height=1.25in,clip,keepaspectratio]{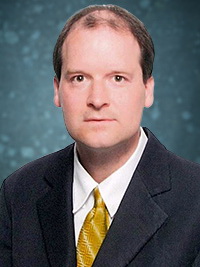}}]{Pascal O.~Vontobel}
	(Fellow, IEEE) received the Diploma degree in electrical engineering in
	1997, the Post-Diploma degree in information techniques in 2002, and the
	Ph.D.\ degree in electrical engineering in 2003, all from ETH Zurich,
	Switzerland.
	
	From 1997 to 2002 he was a research and teaching assistant at the Signal and
	Information Processing Laboratory at ETH Zurich, from 2006 to 2013 he was a
	research scientist with the Information Theory Research Group at
	Hewlett--Packard Laboratories in Palo Alto, CA, USA, and since 2014 he has
	been an Associate Professor at the Department of Information Engineering at
	the Chinese University of Hong Kong. Besides this, he was a postdoctoral
	research associate at the University of Illinois at Urbana--Champaign
	(2002--2004), a visiting assistant professor at the University of
	Wisconsin--Madison (2004--2005), a postdoctoral research associate at the
	Massachusetts Institute of Technology (2006), and a visiting scholar at
	Stanford University (2014). His research interests lie in coding and
	information theory, quantum information processing, data science,
	communications, and signal processing.
	
	Dr.\ Vontobel was an Associate Editor for the {\sc IEEE Transactions on
	Information Theory} (2009--2012), an Awards Committee Member of the IEEE
	Information Theory Society (2013--2014), a Distinguished Lecturer of the IEEE
	Information Theory Society (2014--2015), and an Associate Editor for the {\sc
	IEEE Transactions on Communications} (2014--2017). Moreover, he was a TPC
	co-chair of the 2016 IEEE International Symposium on Information Theory, the
	2018 IEICE International Symposium on Information Theory and its Applications,
	and the 2018 IEEE Information Theory Workshop, along with being the director
	of the 2021 Croucher Summer Course in Information Theory, co-organized
	several topical workshops, and was on the technical program committees of
	many international conferences. Furthermore, he was multiple times a plenary
	speaker at international information and coding theory conferences, he
	received an exemplary reviewer award from the IEEE Communications Society, and
	was awarded the ETH medal for his Ph.D.\ dissertation. He is an IEEE Fellow.
\end{IEEEbiography}
\end{document}